\documentclass[orivec,envcountsect,envcountsame,runningheads]{llncs}
\usepackage[T1]{fontenc}
\usepackage{graphicx}
\usepackage{microtype}
\usepackage{orcidlink}

\usepackage{prelude}

\newif\ifarxiv%
\arxivtrue%

\ifarxiv%
\RequirePackage[paperheight=235mm,paperwidth=155mm,textwidth=12.2cm,textheight=19.3cm,inner=46pt,top=23mm]{geometry}
\makeatletter
\def\@citecolor{blue}%
\def\@urlcolor{blue}%
\def\@linkcolor{blue}%

\makeatother
\else
\RequirePackage[paperheight=235mm,paperwidth=155mm,textwidth=12.2cm,textheight=19.3cm,inner=46pt,top=18mm]{geometry}
\makeatletter
\def\@citecolor{blue}%
\def\@urlcolor{blue}%
\def\@linkcolor{blue}%

\makeatother
\fi%

\title{Partial Reductions for Kleene Algebra\texorpdfstring{\\}{ }with Linear Hypotheses}
\titlerunning{Partial Reductions for KA with Linear Hypotheses}

\author{Liam Chung \orcidlink{0009-0006-1656-3897}\and Tobias Kappé \orcidlink{0000-0002-6068-880X}} % chktex 8
\authorrunning{L. Chung and T. Kappé}
\institute{LIACS, Leiden University, The Netherlands \newline
\email{\{l.w.chung,t.w.j.kappe\}@liacs.leidenuniv.nl}}

\begin{document}

\maketitle

\begin{abstract}
  Kleene algebra $(\KA)$ is an important tool for reasoning about general program equivalences, with a decidable and complete equational theory. 
  However, $\KA$ cannot always prove equivalences between specific programs.
  For this purpose, one adds hypotheses to $\KA$ that encode program-specific knowledge.
  Traditionally, a map on regular expressions called a \emph{reduction} then lets us lift decidability and completeness to these more expressive systems.
  Explicitly constructing such a reduction requires significant labour.
  Moreover, due to regularity constraints, a reduction may not exist for all combinations of expression and hypothesis.

  \smallskip
  We describe an automaton-based construction to mechanically derive reductions for a wide class of hypotheses.
  These reductions can be partial, in which case they yield \emph{partial completeness}: completeness for expressions in their domain.
  This allows us to automatically establish the provability of more equivalences than what is covered in existing work.

\end{abstract}
\section{Introduction}\label{sec:intro}

Program equivalence is a high-value but generally undecidable problem.
In spite of this, there are alternatives for the enterprising computer scientist, by modifying the problem or restricting its domain.
One possibility is to model program behaviour using \emph{regular expressions}, and then use the laws of \emph{Kleene algebra} (\KA) to prove equivalence, typically with respect to a standard, language-based semantics.
Kozen~\cite{kozen-thm} proved that \KA is sound and complete with respect to this semantics.
Due to their connection with finite automata, equivalence of regular expressions is decidable, and so we can decide provable equivalence, too.

Unfortunately, many true equivalences are unprovable by pure \KA, as the standard language semantics is concerned with \emph{propositional equivalences}, i.e., those that hold for \emph{any} interpretation of the programs' primitive statements.
In particular, this excludes the reordering of independent operations.
\KA thus trades expressive power for decidability: one can prove equivalence of propositionally similar programs, while complex equivalences remain unprovable by \KA alone.

At the same time, often there is more program-specific information available, in the form of (in)equations of the form $e \leq f$ for $e,f$ regular expressions. % chktex 36
By introducing these into the proof system as \emph{hypotheses}, new equivalences become provable.
Examples include the properties of Boolean assertions~\cite{kat} as well as the behaviour of network packets~\cite{netkat} and concurrent programs~\cite{ckao}.

By altering \KA to make new equivalences provable, however, the system is no longer sound with respect to regular language semantics.
For example, $\la\lb$ and $\lb\la$ have different languages, though we might add a hypothesis to say that $\la\lb = \lb\la$.
A new semantics is therefore necessary; for example in the case of \emph{Kleene algebra with tests}~\cite{kat} the \emph{guarded string semantics}~\cite{kat-comp-dec} was proposed.
In general, significant work is required to develop such a semantics and subsequently recover soundness, completeness, and decidability --- cf.~\cite{kat-comp-dec,kao,ckao,netkat}.

To make matters worse, not every extension of \KA admits such results.
For instance, \emph{commutativity hypotheses} like the aforementioned $\la\lb = \lb\la$ can be used to encode Post's correspondence problem as an equivalence in \KA with hypotheses, dashing any hope of decidability~\cite{kat,kuznetsov,zhang}.
This is unfortunate, as commutativity hypotheses are especially useful in program equivalence reasoning~\cite{kat}.

In recent literature a meta-theory has developed around \emph{Kleene algebra with hypotheses}.
The cornerstone of this approach is the \emph{hypothesis closure} operation on languages~\cite{doumane}, which gives rise to a sound-by-construction semantics of \KA with additional equations.
The goal then becomes to recover completeness and decidability, by constructing a \emph{reduction}: a map on regular expressions that syntactically implements hypothesis closure, such that the regular expression semantics of a reduced expression coincides with its hypothesis-closed semantics.
Existing work discusses how to construct reductions for some kinds of hypotheses~\cite{ka-eq,doumane}, or combine reductions for different hypotheses~\cite{tools}.

The Achilles heel of this approach is that a reduction must be defined on \emph{all} expressions.
Sometimes, however, hypothesis closure does not preserve regularity, preventing a syntactic realization of its effects; this includes the aforementioned combination of expressions and commutativity hypotheses that encode Post's correspondence problem~\cite{kat}.
In such cases, the method falters, even though regularity might be preserved for \emph{some} expressions.
Our key idea is to define \emph{partial reductions}, which in turn give completeness over their domain of definition.
Using this relaxed notion of reduction, we can recover a limited form of completeness for \KA with hypotheses in a wider variety of scenarios.

The following comprise the core technical contributions of this work:
\begin{itemize}
\item a translation of the hypothesis closure on languages into a construction on automata, for hypotheses of the form $e \leq w$ (where $w$ is a word);
\item a proof that, when this operation produces a finite automaton, it gives rise to a reduction to $\KA$; and
\item a new technique to automatically check and prove equivalence of expressions in \KA with hypotheses in this format, limited to expressions on whose automaton our construction has finite output.
\end{itemize}
If our construction always outputs a finite automaton for a given hypothesis, our method yields a full completeness proof for \KA extended with that hypothesis.

\paragraph{Overview.}
The remainder of this paper is organised as follows.
\cref{sec:background} gives context and necessary definitions.
\cref{sec:patching} discusses and verifies an initial version of our construction, which is improved in \cref{sec:saturation}.
We conclude in \cref{sec:conclusion}.
\ifarxiv%
For brevity's sake, proofs are sketched; full arguments appear in Appendix~\ref{app:proofs}.
\else%
For brevity's sake, proofs are sketched; details appear in the full version~\cite{arxivversion}.
\fi%

\section{Background}\label{sec:background}

Kleene algebra (\KA) models a program as a set of instructions that, depending on various factors (the input, ambient environment) may execute in different ways.
The instructions of the program (e.g., \texttt{print("hi")} or \texttt{x := 5}) are represented as \emph{letters} in a finite \emph{alphabet} $\Sigma = \{ \la, \lb, \dots \}$. % chktex 26 chktex 18 chktex 36
Each possible program execution is a finite sequence of these letters, which creates a \emph{word}, or element of $\Sigma^{*}$.
The semantics of a program is then a set of words, known as a \emph{language}.

By focusing on the sequences of events, \KA abstracts from their specific meaning --- equivalence in \KA is only concerned with \emph{which actions happen, and in what order}.
We refer to this as \emph{propositional equivalence}.
\begin{example}\label{ex:factor-if}
The following programs are propositionally equivalent --- they behave the same regardless of the effects of $\test_1$, $\pa$, $\pb$, and $\pc$.
\begin{prog}[after={}, width=.49\linewidth,
equal height group=A]{}{}
\begin{pseudo}[indent-mark]
\kw{if} $\test_1$ \\+
  $\pa$ \\
  $\pb$ \\-
\kw{else} \\+
$\pc$ \\
$\pb$
\end{pseudo}
\end{prog}
\hfill
\begin{prog}[before={}, width=.49\linewidth,
equal height group=A,label=program:pb-outside]{}{}
\begin{pseudo}[indent-mark]
\kw{if} $\test_1$ \\+
  $\pa$ \\-
\kw{else} \\+
  $\pc$ \\-
$\pb$
\end{pseudo}
\end{prog}
\end{example}

To represent the programs themselves, we use \emph{regular expressions}.
These are formed over the alphabet $\Sigma$, and composed inductively using sequencing ($e \cdot f$), non-deterministic choice ($e + f$) and iteration ($e^{*}$).
Together, these operations can be used to model traditional program composition --- often in conjunction with so called ``tests'' to represent boolean control flow~\cite{kat}.
For instance, the latter program can be modelled as $(\test_1 \cdot \pa + \overline{\test_1} \cdot \pc) \cdot \pb$, where $\test_{1}$ (resp.\ $\overline{\test_1}$) should be read an instruction asserting that $\test_1$ does (resp.\ does not) hold, which aborts the non-deterministic branch on failure.

We formalise these operations, and their connection to languages, here. From here on, we fix a finite alphabet $\Sigma$.

\begin{definition}[Regular expressions, regular language]\label{def:regex}
The set of \emph{regular expressions} $\Exp$ is formed by the following grammar:
\[ e, f ::= 0 \mid 1 \mid \la \in \Sigma \mid e \cdot f \mid e + f \mid e^{*} \]
Every regular expression $e \in \Exp$ can be assigned a regular language $\sem{e} \subseteq \Sigma^{*}$:
\begin{mathpar}
    \sem{0} = \emptyset
    \and
    \sem{1} = \{ \epsilon \}
    \and
    \sem{\la} = \{ \la \}
    \and
    \sem{e + f} = \sem{e} \cup \sem{f}
    \and
    \sem{e \cdot f} = \{ wx : w \in \sem{e}, x \in \sem{f} \}
    \and
    \sem{e^{*}} = \{ w_0 w_1 \cdots w_n : w_0, \dots, w_n \in \sem{e} \}
\end{mathpar}
Here, $ww'$ denotes the concatenation of $w$ and $w'$, and $\epsilon$ denotes the empty word.

A language $L \subseteq \Sigma^{*}$ such that $L = \sem{e}$ for some $e \in \Exp$ is called \emph{regular}.
\end{definition}

We can also represent potential orderings of events using \emph{(finite) automata}. 

\begin{definition}[Automaton]\label{def:automaton}
An \emph{automaton} $\MX$ is a tuple $(X, \rightarrow, x_{\oplus}, x_0)$, in which $X$ is a set of \emph{states}, ${\rightarrow} \subseteq X \times (\Sigma \cup \{\epsilon\}) \times X$ is the \emph{transition relation}, and $x_{\oplus}, x_0 \in X$ are the \emph{final} and \emph{initial state} respectively.
$\MX$ is \emph{finite} when $X$ is; the collection of automata (resp.\ finite automata) is written $\NA^\infty$ (resp.\ $\NA$).

When $(x, a, x') \in {\rightarrow}$ for some $a \in \Sigma \cup \{ \epsilon \}$, we write $x \tr{a} x'$.
Using similar notation, we can then define $\rightarrow^{*}$ as the smallest subset of $X \times {\Sigma^{*}} \times X$ satisfying:
\begin{mathpar}
    \inferrule{%
        x \tr{\la} x'
    }{%
        x \str{\la} x'
    }
    \and
    \inferrule{%
        x \tr{\epsilon} x'
    }{%
        x \str{\epsilon} x'
    }
    \and
    \inferrule{%
        x \str{u} x'
        \and
        x' \str{v} x''
    }{%
        x \str{uv} x'
    }
\end{mathpar}

An element $(x, w, x') \in {\rightarrow^{*}}$ is known as an \emph{$\MX$-trace} for $w$; it is called \emph{accepting} when $x' = x_{\oplus}$, in which case we say $x$ \emph{accepts} $w$.
For a state $x \in X$, \emph{the language accepted} by $x$ is the set of words it accepts:
\[ l_\MX(x) := \{ w \in \words : x \str{w} x_{\oplus} \} \]
The language accepted by $\MX$, denoted $L_\MX$, is the language of its initial state $x_0$.
\end{definition}

The intuitive link between automata and regular expressions as program specifications is legitimised by \emph{Kleene's theorem}~\cite{kleene}, which states that any automaton can be converted into a regular expression with the corresponding language, and vice versa --- thus, they describe the same set of languages.

\begin{theorem}[Kleene's theorem]\label{thm:kleene}
A language $L \in \langs$ is regular if and only if it is the language of a finite automaton. That is, there is an $e \in \Exp$ such that $\sem{e} = L$ if and only if there is a finite automaton $\MX$ such that $L_\MX = L$.
\end{theorem}

Kleene's theorem allows us to move seamlessly between regular expressions and finite automata.
It becomes straightforward to decide equivalence of regular expressions, by simply checking whether their automata are equivalent.

One can also reason algebraically about equivalence of regular expressions, using the rules that make up the proof system \KA~\cite{kozen-thm}.
These include familiar laws such as associativity and commutativity of $+$, in addition to the \emph{fixpoint rules} that dictate the behaviour of the Kleene star.

\begin{definition}[Kleene algebra]\label{def:KA}
The proof system \KA has these axioms:
\begin{mathpar}
  e + 0 = e
  \and
  e + e = e
  \and
  e + f = f + e
  \and
  e + (f + g) = (e + f) + g
  \\
  e \cdot 0 = 0 \cdot e = 0
  \and
  e \cdot 1 = 1 \cdot e = e
  \and
  e \cdot (f \cdot g) = (e \cdot f) \cdot g
  \\
  e \cdot (f + g) = e \cdot f + e \cdot g
  \and
  (e + f) \cdot g = e \cdot g + f \cdot g
  \and
  1 + e \cdot e^{*} = 1 + e^{*} \cdot e = e^{*}
  \\
  e + f \cdot g \leq g \Longrightarrow f^{*} \cdot e \leq g \quad
  \and
  e + f \cdot g \leq f \Longrightarrow e \cdot g^{*} \leq f
\end{mathpar}
where in the last two rules, we use $e \leq f$ as a shorthand for $e + f = f$.

If $\KA \vdash e = f$, we say that $e \equiv f$; equivalently, $\equiv$ is the least congruence on $\Exp$ satisfying the rules of \KA.
We write $e \leqq f$ as a shorthand for $e + f \equiv f$.
\end{definition}

It is not too hard to show that these rules are \emph{sound} for language equivalence.
A landmark result by Kozen~\cite{kozen-thm} shows that \KA is \emph{complete} for the standard language semantics --- i.e., every true $\sem{-}$ equivalence is provable in \KA.

\begin{theorem}[Kozen's theorem]\label{thm:kozen}
\KA is sound and complete w.r.t.\ $\sem{-}$: for all expressions $e,f \in \Exp$, we have that $\sem{e} = \sem{f}$ if and only if  $e \equiv f$.
\end{theorem}

Combining decidability of semantic equivalence for regular expressions with completeness of \KA, it follows that equivalence in \KA is decidable.

The key insight necessary to prove \cref{thm:kozen} is that every automaton can be thought of as a system of equations in \KA, and that solving this system yields an expression representing the language of a state~\cite{kleene,conway-alg-aut,kozen-thm}.
We will leverage this idea heavily in this paper, and so we take a moment to discuss it in depth.

\begin{definition}[Solution to automaton]\label{def:soln-aut}
Let $\MX = (X, \rightarrow, x_{\oplus}, x_0)$ be an automaton.
A \emph{solution} to $\MX$ is a map $s : X \to \Exp$ satisfying the following rules:
\begin{mathpar}
    \inferrule{~}{%
        1 \leqq s(x_\oplus)
    }
    \and
    \inferrule{%
        x \tr{\la} x'
    }{%
        \la \cdot s(x') \leqq s(x)
    }
    \and
    \inferrule{%
        x \tr{\epsilon} x'
    }{%
        s(x') \leqq s(x)
    }
\end{mathpar}
A \emph{least} solution $s$ is a solution where $s(x) \leqq s'(x)$ for all solutions $s'$ and $x \in X$.
\end{definition}

Since least solutions are unique up to $\equiv$, we often speak of \emph{the} least solution to an automaton.
Solutions are connected to the languages in the following way.

\begin{lemma}\label{lem:soln-language}
For any solution $s : X \to \Exp$ to an automaton $\MX$, we have that for every $x \in X$, $l_\MX(x) \sube \sem{s(x)}$.
This inclusion is an equality when $s$ is least.
\end{lemma}

Crucially, for any finite automaton a least solution exists, and we can compute it.
This is typically done with a matrix representation of the transitions of the automaton that we will not cover here; we refer to~\cite{kozen-thm,conway-alg-aut} for more information.

\begin{lemma}\label{lem:least-soln-construct}
For any finite automaton, we can construct a least solution $s_\MX$.
\end{lemma}

\subsection{Hypotheses for Kleene Algebra}\label{subsec:hyp-for-ka}

Because of \KA's focus on propositional equivalence, it concedes significant expressive abilities, leaving many program equivalences on the table.

\begin{restatable}{example}{commWhile}\label{ex:comm-while}
Suppose $\pb$ and $\pa$ are interchangeable in order with no effect on program behaviour, perhaps because they operate on distinct parts of memory; and $\pb$ has no effect on whether $\test_1$ holds.
Then these programs are equivalent:
\begin{prog}[after={}, width=.49\linewidth,
equal height group=B]{}{}
\begin{pseudo}[indent-mark]
$\pb$ \\
\kw{while} $\test_1$ \\+
  $\pa$
\end{pseudo}
\end{prog}
\hfill
\begin{prog}[before={}, width=.49\linewidth,
equal height group=B]{}{}
\begin{pseudo}[indent-mark]
\kw{while} $\test_1$ \\+
  $\pa$ \\-
$\pb$
\end{pseudo}
\end{prog}
\noindent
Still, the corresponding regular expressions have different languages:
\[ \sem{
    \pb \cdot {(\test_1 \cdot \pa)}^{*} \cdot \overline{\test_{1}}
  }
  \neq
  \sem{
    {(\test_1 \cdot \pa)}^{*} \cdot \overline{\test_{1}} \cdot \pb
  } \]
It \emph{is} however possible to prove equivalence of these expressions using \KA, under assumptions about the primitives such as $\pa \cdot \pb = \pb \cdot \pa$,\ $\pb \cdot \test_1 = \test_1 \cdot \pb$,\ $\pb \cdot \overline{\test_1} = \overline{\test_1} \cdot \pb$.
\ifarxiv%
For a proof, see Appendix~\ref{app:example}.
\else%
For a proof, see~\cite[Appendix~A]{arxivversion}.
\fi%
\end{restatable}
This example illustrates that when we represent a program as a sequence of instructions, we forget that some programs have special relationships with others.
In spite of this, \KA can still be used to prove relevant equivalences, provided we incorporate assumptions into our reasoning in the form of (in)equations. % chktex 36

To a significant extent, the situation can be improved by identifying general subclasses of programs that satisfy additional equations, and incorporating those into the theory.
A prime example is \emph{Kleene algebra with tests} ($\KAT$)~\cite{kat} which distinguishes a subset of the atomic programs as Boolean ``tests'', and adds Boolean reasoning principles into the proof system.
$\KAT$ can be given a semantics in terms of \emph{guarded languages}, w.r.t.\ which it is complete and decidable~\cite{kat-comp-dec}.

While this approach is feasible, it is an enormous amount of work to develop such a theory for an individual use-case where we have information that is not indicative of some deeper lack of expressiveness (e.g., Boolean assertions in the case of \KAT), but is rather just useful case-specific information that we want to employ in our reasoning.
Recent literature has developed a meta-theory called \emph{Kleene algebra with hypotheses}~\cite{ka-eq,doumane,ckao,tools}, which parameterises over the specific hypotheses used.
Through this abstraction, we can ask: for which hypotheses are classical results about completeness and decidability recoverable?

\begin{definition}\label{def:hypothesis}
A Kleene algebra \emph{hypothesis} is an (in)equation of two regular expressions, that is, $e \leq f$ for some $e,f \in \Exp$. % chktex 36
A hypothesis is called a \emph{linear hypothesis} if it is of the form $e \leq w$ for some $w \in \Sigma^{*}$.

Given some set of hypotheses $H$, the proof system \KA augmented by the equations in $H$ as additional axioms is called $\KA_H$.
We define $\equiv_H$ as the smallest congruence that satisfies both the axioms of \KA and those asserted by $H$, with the $-_H$ subscript extended to define $\leqq_H$ in a similar manner to $\leqq$.
\end{definition}

\begin{remark}\label{rem:hyp-not-univ}
Letters in a hypothesis $H$ are \emph{not} universally quantified!
That is, if $H = \{ \la\la \leq \la \}$, the $\la$ is a specific letter; $\lb\lb \leqq_H \lb$ is only true when $\la = \lb$.
\end{remark}

We would like to show completeness and decidability of $\KA_H$ for as many sets of hypotheses $H$ as possible.
In doing so, we could be optimistic and hope for a ``silver bullet'' theorem that would work for any set of hypotheses, but that would contradict undecidability of program equivalence.
Indeed, even with commutativity hypotheses, equivalence between certain expressions is undecidable~\cite{kat,kuznetsov,zhang}.
This underscores that we cannot carelessly add axioms into \KA.

\medskip
Setting decidability and completeness aside for a moment, it is important to be clear about what semantics we are working with for $\KA_H$.
Indeed, adding any non-trivial hypotheses to the system will render it unsound with respect to the standard regular language semantics.
For example, if $H = \{ \la\lb \leq \lb\la \}$, then clearly $\la\lb \leqq_H \lb\la$, even though $\sem{\la\lb} \not\subseteq \sem{\lb\la}$.
So, to show completeness of $\KA_H$ (or indeed for that to mean much of anything) we first need a new semantics.
Doumane et al.\ introduced the \emph{hypothesis closure semantics}~\cite{doumane} for this purpose.

\begin{definition}\label{def:hyp-cl-lang}
  We define the \emph{one-step hypothesis closure} for hypotheses $H$:
\[
    H: \langs \to \langs
    \quad
    \text{given by}
    \quad
    L \mapsto L \cup \bigcup \{ u \sem{e} v : u \sem{f} v \subseteq L, e \leq f \in H \}
\]
Where $u \sem{f} v = \{ uwv : w \in \sem{f} \}$, and similarly for $\sem{e}$.

The \emph{hypothesis closure} of $L$, written $H^{*}(L)$, is the smallest superset of $L$ such that $H(H^{*}(L)) = H^{*}(L)$.
We can now define the \emph{hypothesis closure semantics}:
\[
    \sem{-}_H : \Exp \to \langs
    \quad\quad
    \text{given by}
    \quad\quad
    e \mapsto H^{*}(\sem{e})
\]
\end{definition}

For example, the $H$-closure of $\{\la\la\la\la \}$ under $\la \leq \la\la$ is $\{ \la, \la\la, \la\la\la, \la\la\la\la \}$.
$\KA_H$ is sound w.r.t.\ $\sem{-}_{H}$~\cite{doumane}, but completeness is generally much harder.
Methods for some hypotheses were proposed in~\cite{doumane}, which were later developed to recover completeness of $\KAT$ and $\netkat$~\cite{tools}.
Notably, the hypothesis closure semantics is isomorphic to the established semantics for these systems.

\subsection{Reductions to Kleene Algebra}\label{subsec:hyp-cl-intro}
The approach to completeness in existing work~\cite{doumane,tools,ckao,kat-comp-dec,cohen,ka-eq}, has been to leverage \cref{thm:kozen} by realising the semantics of the expanded system in syntax.
Formally, this takes the form of a \emph{reduction}~\cite{tools}.
We call such reductions ``total'', allowing for a ``partial'' version with the totality requirement relaxed.

\begin{definition}[Reduction]\label{def:reduction}
A \emph{total reduction} for a set of hypotheses $H$ is a total map $r: \Exp \to \Exp$ such that for every $e$ in $\Exp$, the following hold:
\begin{mathpar}
    \sem{e}_{H} \sube \sem{r(e)}
    \and
    r(e) \leqq_H e
\end{mathpar}
A \emph{partial reduction} for $H$ is a partial map $r': \Exp \rightharpoonup \Exp$ that satisfies the above two conditions for every $e$ in its domain.
\end{definition}

Partial reductions give us completeness and decidability over their domain; we record this restatement of a well-known fact~\cite{doumane,ckao} below.

\begin{lemma}\label{lem:reduction-complete}
Let $H$ be a set of hypotheses, and $r$ a partial reduction for $H$.
Then $\KA_H$ is complete (w.r.t.\ $\sem{-}_H$) and decidable over the domain of $r$: if $r$ is defined on $e$ and $f$, then $\sem{e}_H = \sem{f}_H$ implies $e \equiv_H f$, and $e \equiv_H f$ is decidable.
In particular, if $r$ happens to be total, the above is true for all $e, f \in \Exp$.
\end{lemma}

\begin{proof}
Note that by definition, for any $f \in \Exp$, $\sem{f} \sube \sem{f}_H$.
Since $r$ is a reduction, $r(e) \leqq_H e$; then by soundness of $\KA_H$, $\sem{r(e)}_H \sube \sem{e}_H$, so $\sem{r(e)} \sube \sem{e}_H$.
Again since $r$ is a reduction, $\sem{e}_H \sube \sem{r(e)}$, so $\sem{e} \sube \sem{r(e)}$. Therefore by \Cref{thm:kozen}, $e \leqq r(e)$.
So if $r$ is a reduction, $\sem{e}_H = \sem{r(e)}$ and $e \equiv_H r(e)$.

Now let $e,f \in \Exp$ in the domain of $r$ be such that $\sem{e}_H = \sem{f}_H$.
Then:
\[ \sem{r(e)} = \sem{e}_H = \sem{f}_H = \sem{r(f)}. \]
We can then apply \cref{thm:kozen} to find that $r(e) \equiv r(f)$ --- these expressions are equivalent by \KA alone.
Since $\KA_H$ has all of the rules of \KA, it follows that $r(e) \equiv_H r(f)$.
Lastly, we can then combine this with the other requirement for partial reductions to conclude that $e \equiv_H r(e) \equiv_H r(f) \equiv_H f$.

As for decidability, note that the above (in combination with soundness) tells us that to decide $e \equiv_H f$ is to decide whether $\sem{r(e)} = \sem{r(f)}$, and equivalence of regular expressions is decidable.
\end{proof}

Various works have defined reductions tailored to specific hypotheses.
For example, in~\cite{kao} (see also~\cite{ka-eq}) the ``contraction'' hypothesis $\la \leq \la\la$ is realised by applying ``transitive closure'' to automata, as detailed below.

\begin{example}\label{ex:contraction-trcl}
Let $H = \{ \la \leq \la\la \}$.
This hypothesis can quite naturally be seen as transitive closure at the automaton level.
Any two states connected by a sequence of two \la-transitions should also be connected by just one \la-transition:
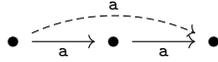
\[\begin{tikzcd} % https://q.uiver.app/#q=WzAsMyxbMCwwLCJcXGJ1bGxldCJdLFsxLDAsIlxcYnVsbGV0Il0sWzIsMCwiXFxidWxsZXQiXSxbMCwxLCJcXGxhIiwyXSxbMSwyLCJcXGxhIiwyXSxbMCwyLCJcXGxhIiwwLHsiY3VydmUiOi0yLCJzdHlsZSI6eyJib2R5Ijp7Im5hbWUiOiJkYXNoZWQifX19XV0=
	\bullet & \bullet & \bullet
	\arrow["\la"', from=1-1, to=1-2]
	\arrow["\la", curve={height=-12pt}, dashed, from=1-1, to=1-3]
	\arrow["\la"', from=1-2, to=1-3]
\end{tikzcd}\]
For any expression $e$, converting $e$ to an automaton, applying this ``transitive closure'', and converting back to regular expression constitutes a total reduction.
\end{example}

Indeed many reductions are defined through automata in some fashion: we can think of a hypothesis $e \leq f$ as ``if we can move between two states with any word from $f$, we should be able to do the same while reading any word from $e$''.
In contrast, working directly on expressions to define a reduction is difficult, because their inductive structure is of little help.

We therefore approach the general reduction problem by converting a regular expression $g$ into an automaton, applying a closure construction there, and then converting the result back into an expression $r(g)$.
We can then use algebraic representations of the automata to prove that $r(g) \leqq_H g$.

Of course, for a reduction of $g$ to be feasible at all, $\sem{g}_H$ needs to be a regular language; if it is not, then there is no hope of finding a sensible expression for $r(g)$.
This can happen for commutativity hypotheses.

\begin{example}\label{ex:ba-ab-undec}
Let $H := \{ \la\lb \leq \lb\la \}$.
Then $\sem{{(\la\lb)}^{*}}_H$ cannot be regular, because otherwise $\sem{{(\la\lb)}^{*}}_H \cap \sem{\la^{*}\lb^{*}}$ would also be regular, and the latter is precisely $\{ \la^{n}\lb^{n} : n \in \mathbb{N} \}$, which is a well-known non-regular (but context free) language.
\end{example}

While a total reduction, which gives a decision procedure and completeness all in one, is undoubtedly the best case scenario, it fails for certain combinations of hypothesis and expression, as in \cref{ex:ba-ab-undec}.
However, this also means giving up on many useful equivalences, even if all of the (hypothesis-closed) languages involved are regular.
Our main innovation is a technique to construct reductions for \emph{some} expressions w.r.t.\ hypotheses $H$, even if it might fail for others.
We thereby widen the perspective beyond systems we can show to be decidable and complete via a total reduction.

\section{Partial Reduction: Patching}\label{sec:patching}

We now define a notion of hypothesis closure on automata.
When paired with standard algorithms converting between regular expressions and automata (e.g.~\cite{thompson}), this operation yields a partial reduction for the given hypothesis.

\begin{restatable}{assume}{bigAssumption}
    From here on, unless specified, we use a hypothesis set $H = \{ e \leq w \}$, and an automaton $\MX := (X, \rightarrow, x_{\oplus}, x_{0})$.
\end{restatable}

We first state the desired effect on one state.
Recall that the hypothesis $e \leq w$ enforces that if there is a word in $L$ with $w$ as a subword --- i.e., $uwv \in L$ --- then we can replace $w$ with any word from $\sem{e}$, giving $uw'v \in H(L)$ for each $w' \in \sem{e}$.
To mimic this effect on automata, we will seek out states where we can read $w$ to reach some other state.
We then add an alternate path: instead of following $w$, step to a new automaton recognising $\sem{e}$, and from the accepting state of that automaton, go back to any of the states we could have gone to with $w$.
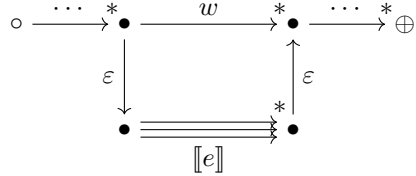
\[\begin{tikzpicture}[baseline=(current bounding box.center)]

        \node (s0) {$\circ$};
        \node[right=of s0] (s1) {$\bullet$};
        \node[right= 1.8cm of s1] (s2) {$\bullet$};
        \node[right=of s2] (s3) {$\oplus$};
        \node[below=of s1] (s4) {$\bullet$};
        \node[below=of s2] (s5) {$\bullet$};

        \draw[->, decorate] (s0) edge
          node[above] {$\cdots$}
          node[above,pos=1.05] {$*$}
        (s1);

        \draw[->] (s1) edge
          node[above] {$w$}
          node[above,pos=1.02] {$*$}
        (s2);
        \draw[->] (s2) edge
          node[above] {$\cdots$}
          node[above,pos=1.02] {$*$}
        (s3);
        \draw[->] (s4) edge (s5);
        \draw[->,transform canvas={yshift=-.1cm}] (s4) edge
          node[below] {$\sem{e}$}
        (s5);
        \draw[->,transform canvas={yshift=.1cm}] (s4) edge
          node[above,pos=1.02] {$*$}
        (s5);
        \draw[->] (s1) edge node[left] {$\epsilon$} (s4);
        \draw[->] (s5) edge node[right] {$\epsilon$} (s2);

\end{tikzpicture}\]

\begin{definition}[Patching]\label{def:patching}
Let $x \in X$, let $\MZ := (Z, \rightarrow_\MZ, z_\oplus, z_0)$ be an automaton.
We define the automaton $\MX\xpatch$ (read: $\MX$ with $\MZ$ $w$-\emph{patched at $x$}) as $(X \cup Z, \rightarrow_P, x_\oplus, x_0)$,\footnote{Without loss of generality, the state sets $X$ and $Z$ are assumed to be disjoint.} where $\rightarrow_P$ is the smallest subset of $(X \cup Z) \times (\Sigma \cup \{\epsilon\}) \times (X \cup Z)$ satisfying the following rules, where $a \in \Sigma \cup \{ \epsilon \}$:
\begin{mathpar}
    \inferrule{%
        y \tr[\MX]{a} y'
    }{%
        y \tr[P]{a} y'
    }
    \and
    \inferrule{%
        y \tr[\MZ]{a} y'
    }{%
        y \tr[P]{a} y'
    }
    \and
    \inferrule{~}{%
        x \tr[P]{\epsilon} z_0
    }
    \and
    \inferrule{
        x \str[\MX]{w} x'
    }{%
        z_\oplus \tr[P]{\epsilon} x'
    }
\end{mathpar}
\end{definition}

Patching mimics hypothesis closure w.r.t.\ $H$ on automata, local to one state.

\begin{example}\label{ex:patching-intro}
As seen in \cref{ex:contraction-trcl}, an automaton construction was used in~\cite{kao} to realise the hypothesis $\la \leq \la\la$.
The hypothesis replaces any instance of the subword $\la\la$ with $\la$, and in automaton terms, any instance of two $\la$-transitions can also be traversed by a single $\la$-transition: a kind of ``transitive closure''.
As can be seen below, patching replicates this notion, up to $\epsilon$-removal:
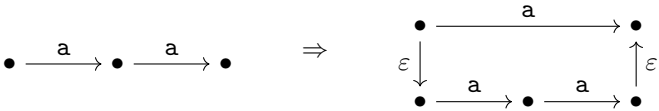
\[
    \begin{tikzpicture}[baseline=(current bounding box.center)]
        \node (s1) {$\bullet$};
        \node[right=of s1] (s2) {$\bullet$};
        \node[right=of s2] (s3) {$\bullet$};
        \draw[->] (s1) edge node[above] {$\la$} (s2);
        \draw[->] (s2) edge node[above] {$\la$} (s3);
    \end{tikzpicture}
    \quad\quad
    \Rightarrow
    \quad\quad
    \begin{tikzpicture}[baseline=(current bounding box.center)]
        \node (s1) {$\bullet$};
        \node[right=of s1] (s2) {$\bullet$};
        \node[right=of s2] (s3) {$\bullet$};
        \node[above=6mm of s1] (s4) {$\bullet$};
        \node[above=6mm of s3] (s5) {$\bullet$};
        \draw[->] (s1) edge node[above] {$\la$} (s2);
        \draw[->] (s2) edge node[above] {$\la$} (s3);
        \draw[->] (s1) edge node[left] {$\epsilon$} (s4);
        \draw[->] (s4) edge node[above] {$\la$} (s5);
        \draw[->] (s5) edge node[right] {$\epsilon$} (s3);
    \end{tikzpicture}
\]
\end{example}

Just like the one-step language closure operation is iterated to define language closure, multiple patching operations may be necessary to achieve the desired closure --- either to the newly added states, or to states already in $\MX$.
The latter happens when, after patching one state, opportunities to patch other states become apparent, as shown in the example below.

\begin{example}\label{ex:multiple-patches}
Let $H = \{ \la \leq \lb\la \}$, and consider the expression $\lb\lb\la$.
Then, applying our patching construction one time, we obtain an automaton whose language is not yet closed: the output automaton needs to accept $\la$, as $\la \in \sem{\lb\lb\la}_{\la \leq \lb\la}$. So we patch again, obtaining the desired output:
% https://q.uiver.app/#q=WzAsMTgsWzAsMCwiXFxjaXJjIl0sWzAsMSwiXFxidWxsZXQiLFswLDYwLDYwLDFdXSxbMCwyLCJcXGJ1bGxldCIsWzAsNjAsNjAsMV1dLFsyLDAsIlxcY2lyYyIsWzAsNjAsNjAsMV1dLFsyLDEsIlxcYnVsbGV0IixbMCw2MCw2MCwxXV0sWzIsMiwiXFxidWxsZXQiXSxbMiwzLCJcXG9wbHVzIl0sWzMsMSwiXFxidWxsZXQiLFswLDYwLDYwLDFdXSxbMywyLCJcXGJ1bGxldCIsWzAsNjAsNjAsMV1dLFs1LDAsIlxcY2lyYyJdLFs1LDEsIlxcYnVsbGV0Il0sWzUsMiwiXFxidWxsZXQiXSxbNSwzLCJcXG9wbHVzIl0sWzYsMSwiXFxidWxsZXQiXSxbNiwyLCJcXGJ1bGxldCJdLFs3LDAsIlxcYnVsbGV0Il0sWzcsMSwiXFxidWxsZXQiXSxbMCwzLCJcXG9wbHVzIixbMCw2MCw2MCwxXV0sWzAsMSwiXFxsYiIsMl0sWzEsMiwiXFxsYiIsMix7ImNvbG91ciI6WzAsNjAsNjBdfSxbMCw2MCw2MCwxXV0sWzMsNCwiXFxsYiIsMix7ImNvbG91ciI6WzAsNjAsNjBdfSxbMCw2MCw2MCwxXV0sWzQsNSwiXFxsYiIsMl0sWzUsNiwiXFxsYSIsMl0sWzQsNywiXFxlcHNpbG9uIiwwLHsiY29sb3VyIjpbMCw2MCw2MF19LFswLDYwLDYwLDFdXSxbOCw2LCJcXGVwc2lsb24iXSxbNyw4LCJcXGxhIiwwLHsiY29sb3VyIjpbMCw2MCw2MF19LFswLDYwLDYwLDFdXSxbOSwxMCwiXFxsYiIsMl0sWzEwLDExLCJcXGxiIiwyXSxbMTEsMTIsIlxcbGIiLDJdLFsxNCwxMiwiXFxlcHNpbG9uIl0sWzEwLDEzLCJcXGVwc2lsb24iXSxbMTMsMTQsIlxcbGEiLDJdLFs5LDE1LCJcXGVwc2lsb24iXSxbMTUsMTYsIlxcbGEiXSxbMTYsMTQsIlxcZXBzaWxvbiJdLFsyLDE3LCJcXGxhIiwyLHsiY29sb3VyIjpbMCw2MCw2MF19LFswLDYwLDYwLDFdXSxbMTksMjEsIiIsMCx7InNob3J0ZW4iOnsic291cmNlIjo0MCwidGFyZ2V0Ijo0MH19XSxbMjUsMjcsIiIsMCx7InNob3J0ZW4iOnsic291cmNlIjo0MCwidGFyZ2V0Ijo0MH19XV0=
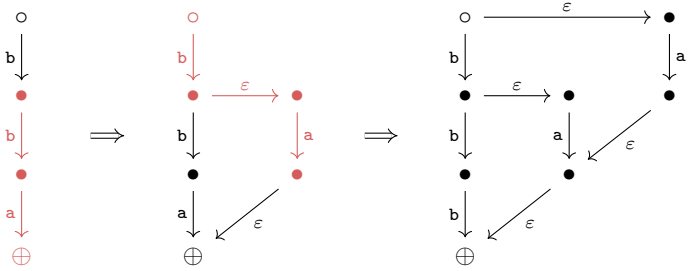
\[\begin{tikzcd}
	\circ && \textcolor{rgb,255:red,214;green,92;blue,92}{\circ} &&& \circ && \bullet \\
	\textcolor{rgb,255:red,214;green,92;blue,92}{\bullet} && \textcolor{rgb,255:red,214;green,92;blue,92}{\bullet} & \textcolor{rgb,255:red,214;green,92;blue,92}{\bullet} && \bullet & \bullet & \bullet \\
	\textcolor{rgb,255:red,214;green,92;blue,92}{\bullet} && \bullet & \textcolor{rgb,255:red,214;green,92;blue,92}{\bullet} && \bullet & \bullet \\
	\textcolor{rgb,255:red,214;green,92;blue,92}{\oplus} && \oplus &&& \oplus
	\arrow["\lb"', from=1-1, to=2-1]
	\arrow["\lb"', color={rgb,255:red,214;green,92;blue,92}, from=1-3, to=2-3]
	\arrow["\epsilon", from=1-6, to=1-8]
	\arrow["\lb"', from=1-6, to=2-6]
	\arrow["\la", from=1-8, to=2-8]
	\arrow[""{name=0, anchor=center, inner sep=0}, "\lb"', color={rgb,255:red,214;green,92;blue,92}, from=2-1, to=3-1]
	\arrow["\epsilon", color={rgb,255:red,214;green,92;blue,92}, from=2-3, to=2-4]
	\arrow[""{name=1, anchor=center, inner sep=0}, "\lb"', from=2-3, to=3-3]
	\arrow[""{name=2, anchor=center, inner sep=0}, "\la", color={rgb,255:red,214;green,92;blue,92}, from=2-4, to=3-4]
	\arrow["\epsilon", from=2-6, to=2-7]
	\arrow[""{name=3, anchor=center, inner sep=0}, "\lb"', from=2-6, to=3-6]
	\arrow["\la"', from=2-7, to=3-7]
	\arrow["\epsilon", from=2-8, to=3-7]
	\arrow["\la"', color={rgb,255:red,214;green,92;blue,92}, from=3-1, to=4-1]
	\arrow["\la"', from=3-3, to=4-3]
	\arrow["\epsilon", from=3-4, to=4-3]
	\arrow["\lb"', from=3-6, to=4-6]
	\arrow["\epsilon", from=3-7, to=4-6]
	\arrow[between={0.4}{0.6}, Rightarrow, from=0, to=1]
	\arrow[between={0.4}{0.6}, Rightarrow, from=2, to=3]
\end{tikzcd}\]
For each step, the part to be patched next is \textcolor{q-red}{highlighted}.
Clearly, we can force patching to be necessary up to $n$ times if we like, by using $H = \{ \la \leq \lb\la\}$ on an automaton for the expression $\lb^n\la$. 
\end{example}

Of course, repeated patching opens up the possibility that we never stop patching, despite the fact that the constructed automaton must be finite to extract an expression.
Indeed, no correct automaton closure construction (in the sense that it outputs an automaton for the closed language) always preserves finite automata, since hypothesis closure does not preserve regularity (cf.\ \cref{ex:ba-ab-undec}).

While the construction we seek to define may need to iterate patching to reach a correct output, it should terminate where possible.
Therefore, it needs a criterion to judge if a state should be patched, and indeed if its work on the automaton as a whole is done.
To this end, we recall the \emph{Brzozowski derivative}, as well as its generalisation, which operates on languages.

\begin{definition}[Brzozowski derivative]\label{def:brz-derivs}
Let $L$ be a language, and $w \in \words$ a word.
The \emph{Brzozowski derivative} of $L$ with respect to $w$ is defined as follows:
\[ w\inv L := \{ u \in \words : wu \in L \}. \]
Let $L,M$ be languages. The \emph{generalised Brzozowski derivative} (also known as the \emph{left residual}) of $L$ with respect to $M$ is defined as follows:
\[ M\inv L := \bigcap_{w \in M} w\inv L = \{ v \in \words : wv \in L, \text{ for all } w \in M \}. \]
\end{definition}

The following characterisation can then be used to check whether an automaton requires any more patching --- that is, if the languages of each of the automaton's states are closed with respect to the hypothesis.

\begin{restatable}{lemma}{termCond}\label{lem:term-cond}
For any automaton $\MX$,
\[ 
    \forall x \in X,\ w\inv l_\MX(x) \sube \sem{e}\inv l_\MX(x)
    \text{ if and only if }
    \forall x \in X,\ l_\MX(x) = H(l_\MX(x)).
\]
\end{restatable}

Note the location of the quantifiers: the language containment holds at \emph{all} states if and only if \emph{all} states are hypothesis closed.
This equivalence relates hypothesis-closure of the language of each state to a containment of one Brzozowski derivative in another.
Since regular languages are closed under (generalised) Brzozowski derivatives, the latter is computable.

We now define the first version of our construction on automata.
First we patch all states in the automaton as needed according to the criterion from \cref{lem:term-cond}.
This will imitate the action of $H$ at every state in $\MX$ at once.

\begin{definition}[$T_0$]\label{def:T0}
Write $X = \{ x_0, \dots, x_{n-1} \}$; we define $T_0(\MX) = \MX_n$, where
\begin{mathpar}
    \MX_0 = \MX
    \and
    \MX_{i+1} = 
        \begin{cases}
        \MX_i & w\inv l_\MX(x_i) \sube \sem{e}\inv l_\MX(x_i) \\
        \MX_i\patch{x_i} & \text{otherwise}
        \end{cases}
\end{mathpar}
\end{definition}

\begin{remark}\label{rem:well-defined-T0}
The languages used to check whether $T_0$ pastes onto a state $x_i$ are from $\MX$, not the intermediate automata $\MX_i$.
Thus all patches occur simultaneously, and the output of $T_0$ is independent of the order states are considered.
\end{remark}

One application of $T_0$ to $\MX$ introduces new states and grows the language of existing ones, potentially necessitating further patching.
Consequently, $T_0$ must be iterated (potentially infinitely often) until a hypothesis-closed automaton arises.
To rigorously define this process and formalise the notion of ``building'' the output of the construction, we introduce an order on automata.

\begin{definition}[Automaton order]\label{def:aut-order}
Let $\sqsube$ be the order on automata where
\[ (X, \rightarrow, x_\oplus, x_0) \sqsube (X', \rightarrow', x_\oplus', x_0') \text{ iff } X \sube X', {\rightarrow} \sube {\rightarrow'}, x_0 = x_0', x_{\oplus} = x_{\oplus}'. \]
\end{definition}
\noindent
Clearly, if $\MX \sqsube \MX'$, then $l_{\MX}(x) \sube l_{\MX'}(x)$ for every state $x$ in $\MX$.

When equipped with $\sqsube$, automata form an $\omega$-complete partial order, meaning any chain $\MX_0 \sqsube \MX_1 \sqsube \cdots$ must have a \emph{limit} $\MX_*$, which is the $\sqsube$-least automaton such that $\MX_i \sqsube \MX_*$ for each $i$.
This allows us to define the output of the construction as the limit of a chain --- with possibly infinitely many states.

\begin{definition}[$T_0^*$]\label{def:T0star}
 $T_0^{*}(\MX)$ is the limit of $\MX \sqsube T_0(\MX) \sqsube T_0^2(\MX) \sqsube \cdots$.
\end{definition}
\noindent
With standard constructions, $T_0^*$ lifts to a partial function on regular expressions.

\begin{definition}[Candidate partial reduction: $r_0$]\label{def:rZ}
  Given a regular expression $g$, let $\MX$ be its finite automaton.
  If $T_0^*(\MX)$ is finite, then $r_0(g)$ is the regular expression for the initial state of $T_0^{*}(\MX)$.
  Otherwise, it is undefined.
\end{definition}

For $r_0$ to be a partial reduction, \cref{def:reduction} tells us we need to show two things when it is defined on a regular expression $g$.
First, $\sem{g}_H$ should be contained in $\sem{r_0(g)}$; this boils down to showing that the $H$-closure of the language of an automaton $\MX$ is contained in the (plain) language of $T_0^*(\MX)$.

\begin{restatable}{lemma}{TZcorrect}\label{lem:T0correct}
For every $x \in X$, it holds that $H^{*}(l_\MX(x)) \sube l_{T_0^{*}(\MX)}(x)$.
\end{restatable}
\begin{proof}[Proof sketch]
  One can more easily show the inclusion at one step: for every $x \in X$, $H(l_{\MX(x)}) \sube l_{T_0}(\MX)(x)$.
  This requires taking a word of the form $u w v \in l_{\MX}(x)$ and verifying that $u w' v \in l_{T_0(\MX)}(x)$ for any $w' \in \sem{e}$, which is simply a matter of applying the criterion in \cref{lem:term-cond} and tracing the patch as it is constructed.
\end{proof}

To validate the second part of \cref{def:reduction}, we must show that $r_0(g) \leqq_H g$.
We will achieve this using the systems of equations for the automata corresponding to these expressions, $\MX$ and $T_0^{*}(\MX)$.
To this end, we first extend the machinery around least solutions to $\KA_H$.

\begin{definition}\label{def:Hsoln}
An \emph{$H$-solution} to $\MX$ is a map $s : X \to \Exp$ satisfying the rules from \cref{def:soln-aut}, but with $\leqq_H$ instead of $\leqq$.
\end{definition}
\begin{restatable}{lemma}{solnAutExtended}\label{lem:solnAutExtended}
The (plain) least solution $s_\MX$ to $\MX$ from \cref{lem:least-soln-construct} is also an $H$-solution to $\MX$, and in fact it is \emph{least} among all $H$-solutions to $\MX$.
\end{restatable}
\begin{proof}[Proof sketch]
The first claim follows because $\leqq_H$ is weaker than $\leqq$.
One can then use the same construction used for \cref{lem:least-soln-construct} to obtain a least $H$-solution, but the expression constructed is the same as that obtained for the least solution.
\end{proof}

Of course, (least) $H$-solutions may not be (least) solutions in the sense of \cref{def:soln-aut}.
\cref{lem:solnAutExtended} provides leverage, relating $H$-solutions of an automaton to its (plain) least solution.
To use this, we first show that $T_0^{*}$ does not (up to $H$) perturb a solution of $\MX$ --- we only need to include the new states.

\begin{restatable}{lemma}{extendSolnTZ}\label{lem:extend-soln-T0}
Suppose $T_0^{*}(\MX)$ is finite, and $s$ is a solution to $\MX$.
We can construct an $H$-solution $s^{*}$ to $T_0^{*}(\MX)$, which moreover agrees with $s$ on all $x \in X$.
\end{restatable}
\begin{proof}[Proof sketch]
Since $T_0^{*}(\MX)$ is finite, it is $T_0$ applied to $\MX$ finitely many times.
We prove the claim just for $T_0(\MX)$; the main claim then follows by induction.
Let $s_{\MZ}$ be the least solution to $\MZ$.
$T_0(\MX)$ has $\MZ_1, \ldots \MZ_n$ (all copies of $\MZ$) with state sets $Z_1, \dots, Z_n$ $w$-patched on at states $x_1, \ldots x_n$.
One defines $s^{*}$ as follows.
\[
s^{*}(y) = \begin{cases}
     s(y) & y \in X \\
     s_{\MZ}(y) \cdot r_i & y \in Z_i
\end{cases}
\quad\quad\text{where}\quad\quad
r_i := \sum_{x_i \str{w} x' }s(x')
\]
For a state in $\MX$, $s^*$ just looks at $s$; for a state in $\MZ_i$, it uses $s_{\MZ}$ while accounting for the fact that we need to re-enter $\MX$, by composing with $r_i$.

The majority of the equations required for $s^{*}$ to be an $H$-solution to $T_0^{*}(\MX)$ (i.e., the ones resulting from the first two rules in \cref{def:patching}) follow by definition.
The only ones that remain correspond to $\epsilon$-transitions out to, and back from, the patched on automata (the last two rules in \cref{def:patching}).
We first look at equations resulting from the fourth rule, where an $\epsilon$-transition from some $\MZ_i$ goes back into $\MX$.
Here, we must prove that $s^{*}(x') \leqq s^{*}(z_{\oplus}^{i})$ where $x'$ is one of the states where $\MZ_i$ transitions back to $\MX$.
This is proved as follows:
\[ s^{*}(x') \equiv 1 \cdot s^{*}(x') \leqq s_{\MZ}(z_{\oplus}^i) \cdot s^{*}(x') \leqq s_{\MZ}(z_{\oplus}^i) \cdot r_i = s^{*}(z_{\oplus}^i) \]
For the second-to-last step, we use the definition of $r_i$ to observe that $s^{*}(x') = s(x') \leqq r_i$ because $x_i \str{w} x'$.
All other steps follow by the axioms of $\KA$, the fact that $z_{\oplus}^i$ is the accepting state of $\MZ_i$ with $s_\MZ$ as solution, and the definition of $s^{*}$.

Lastly, each equation corresponding to an $\epsilon$-transition from $\MX$ out to some $\MZ_i$ is of the form $s^{*}(z_{0}^i) \leqq s^{*}(x_i)$, which is proved as follows:
\begin{equation}\label{eq:extend-soln-change-next}
  s^{*}(z_0^i) \equiv s_{\MZ}(z_0^i) \cdot r_i \equiv e \cdot r_i \leqq_H w \cdot r_i \leqq s(x_i) = s^{*}(x_i).
\end{equation}
For the second-to-last step, i.e., $w \cdot r_i \leqq s(x_i)$, it suffices to show that for all $x'$ such that $x_i \str{w} x'$, $w \cdot s(x') \leqq s(x_i)$.
This follows directly from the definition of solution, applied in an induction on the construction of $x_i \str{w} x'$.

All other steps in~\eqref{eq:extend-soln-change-next} follow by the definition of $s^{*}$, the definition of $\MZ_i$ as representing $e$, and the (only) hypothesis in $H$. 
The use of $H$ in the middle is the reason that $s^{*}$ is an $H$-solution of $T_0(\MX)$, not necessarily a plain solution.
\end{proof}

Since the least solution of $\MX$ is part of an $H$-solution for $T_0^{*}(\MX)$, it is above the least solution for $T_0^{*}(\MX)$, up to $H$.
Because $r_0(g)$ is created from the least solution (and thus the least $H$-solution) to $T_0^*(\MX)$, we can relate $g$ to $r_0(g)$.

\begin{theorem}\label{thm:partialRedTZ}
$r_0$ is a partial reduction.
\end{theorem}
\begin{proof}
  We show that for every $g \in \Exp$ on which $r_0$ is defined, we have $\sem{g}_{H} \sube \sem{r_0(g)}$ and $r_0(g) \leqq_H g$.
  The former follows from \cref{lem:T0correct}:
  \[ \sem{g}_{H} = H^{*}(L_{\MX}) \sube L_{T_0^{*}(\MX)}  = \sem{r_0(g)} \]
  For the latter, let $s$ be the least solution to $\MX$, the automaton representing $g$; by \Cref{thm:kozen}, $s(x_0) \equiv g$.
  Let $s'$ be the least solution to $T_0^{*}(\MX)$; by definition, $r_0(g) = s'(x_0)$.
  \cref{lem:extend-soln-T0} extends $s$ to an $H$-solution $s^*$ for $T_0^{*}(\MX)$.
  Since $s'$ is the \emph{least} solution to $T_0^{*}(\MX)$, it is a least $H$-solution, allowing us to derive:
  \[ r_0(g) = s'(x_0) \leqq_{H} s^*(x_0) = s(x_0) \equiv g \qedhere \]
\end{proof}

With our partial reduction in hand, we can now conclude decidability and completeness for the cases where it is defined, by way of \cref{lem:reduction-complete}.

\begin{corollary}\label{cor:partcompTZ}
  Let $g,h$ be expressions for which $r_0$ is defined.
  If $\sem{g}_H = \sem{h}_H$, then $g \equiv_H h$; moreover, the latter is decidable.
  Thus, if $r_0$ is known to output finite automata for all expressions in $\Exp$, then $\KA_H$ is complete and decidable.
\end{corollary}

In the next section we will improve on $T_0$, so that its output is finite in more situations.
This will allow for wider domains of reduction for some hypotheses.
However, we can already apply $T_0^{*}$ to recover some known results.

For contraction hypotheses $\la \leq \la\la$ (so-named in~\cite{ckao}), $T_0^{*}$ always outputs finitely.
As such, the reduction given in op.\ cit., and the completeness achieved as a result, are recoverable using this method.
In~\cite[Lemma~4.35]{ckao}, linear hypotheses (called ``grounded'' there) are also shown to lift to a concurrent setting.
Finite output generalises to other ``contraction hypotheses'' $\la \leq w$, as $T_0$ emulates the ``descendent construction'' described in the proof of~\cite[Theorem~4.1.2]{descendent}.

While $T_0^{*}$ only operates on one hypothesis at a time, we can still make use of it for ``independent'' sets of hypotheses, where closure with respect to all of them at once is the same as closure w.r.t.\ each hypothesis individually, in some order~\cite{ckao,tools}.
Then we can use $T_0^{*}$ to construct a reduction from $\KA_H$ down to $\KA$, removing one hypothesis at a time.
This is easily shown to be the case for a set of contraction hypotheses, for example $H := \{ \la \leq \la\la,\ \lb \leq \lb\lb,\ \ldots \}$.

\section{Improved Partial Reduction: Saturation}\label{sec:saturation}

\cref{cor:partcompTZ} is predicated on the domain of $r_0$, and by extension on whether $T_0^*$ outputs a finite automaton when fed the automata for $g$ and $h$, but does not tell us when this might be the case.
While we cannot expect $T_0^{*}$ to \emph{always} produce a finite automaton (per \cref{ex:ba-ab-undec}), we wish for it to do so \emph{when possible} --- i.e., when the hypothesis closure of the language in question is regular.
Unfortunately, there are cases where $T_0^*$ unnecessarily outputs an infinite automaton.

\begin{example}\label{ex:sat-motiv}
Let $H := \{ \lb\la \leq \la \}$, and consider the expression $\la$.
We can easily calculate that the desired output expression is $\lb^{*}\la$, but $T_0^{*}$ proceeds infinitely:
% https://q.uiver.app/#q=WzAsMTcsWzcsMCwiXFxjaXJjIl0sWzgsMCwiXFxvcGx1cyJdLFs0LDAsIlxcb3BsdXMiXSxbMywwLCJcXGNpcmMiXSxbNywxLCJcXGJ1bGxldCJdLFs4LDEsIlxcYnVsbGV0Il0sWzksMSwiXFxidWxsZXQiXSxbOCwyLCJcXGJ1bGxldCJdLFs5LDIsIntcXGNvbG9ye3EtcmVkfSBcXGJ1bGxldH0iXSxbMTAsMiwie1xcY29sb3J7cS1yZWR9IFxcYnVsbGV0fSJdLFszLDEsIlxcYnVsbGV0Il0sWzQsMSwie1xcY29sb3J7cS1yZWR9IFxcYnVsbGV0fSJdLFs1LDEsIlxcYnVsbGV0IixbMCw2MCw2MCwxXV0sWzEsMCwiXFxvcGx1cyIsWzAsNjAsNjAsMV1dLFswLDAsIlxcY2lyYyIsWzAsNjAsNjAsMV1dLFsxLDFdLFs1LDBdLFswLDEsIlxcbGEiXSxbMywyLCJcXGxhIl0sWzAsNCwiXFxlcHNpbG9uIiwyXSxbNCw1LCJcXGxiIiwyXSxbNSw2LCJcXGxhIiwyXSxbNiwxLCJcXGVwc2lsb24iLDJdLFs1LDcsIlxcZXBzaWxvbiIsMl0sWzcsOCwiXFxsYiIsMl0sWzgsOSwiXFxsYSIsMix7ImNvbG91ciI6WzAsNjAsNjBdfSxbMCw2MCw2MCwxXV0sWzksNiwiXFxlcHNpbG9uIiwyXSxbMTQsMTMsIlxcbGEiLDAseyJjb2xvdXIiOlswLDYwLDYwXX0sWzAsNjAsNjAsMV1dLFsxMywxNSwiIiwwLHsic3R5bGUiOnsiYm9keSI6eyJuYW1lIjoibm9uZSJ9LCJoZWFkIjp7Im5hbWUiOiJub25lIn19fV0sWzMsMTAsIlxcZXBzaWxvbiIsMl0sWzEyLDIsIlxcZXBzaWxvbiIsMl0sWzEwLDExLCJcXGxiIiwyXSxbMTEsMTIsIlxcbGEiLDIseyJjb2xvdXIiOlswLDYwLDYwXX0sWzAsNjAsNjAsMV1dLFsxNiwxMiwiIiwyLHsic3R5bGUiOnsiYm9keSI6eyJuYW1lIjoibm9uZSJ9LCJoZWFkIjp7Im5hbWUiOiJub25lIn19fV0sWzI4LDI5LCJUXzAiLDAseyJzaG9ydGVuIjp7InNvdXJjZSI6NDAsInRhcmdldCI6NDB9fV0sWzMzLDE5LCJUXzAiLDAseyJzaG9ydGVuIjp7InNvdXJjZSI6NDAsInRhcmdldCI6NDB9fV1d
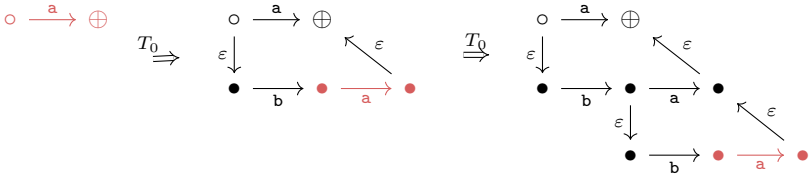
\[\begin{tikzcd}[sep=scriptsize]
	\textcolor{rgb,255:red,214;green,92;blue,92}{\circ} & \textcolor{rgb,255:red,214;green,92;blue,92}{\oplus} && \circ & \oplus & {} && \circ & \oplus \\
	& {} && \bullet & {{\color{q-red} \bullet}} & \textcolor{rgb,255:red,214;green,92;blue,92}{\bullet} && \bullet & \bullet & \bullet \\
	&&&&&&&& \bullet & {{\color{q-red} \bullet}} & {{\color{q-red} \bullet}}
	\arrow["\la", color={rgb,255:red,214;green,92;blue,92}, from=1-1, to=1-2]
	\arrow[""{name=0, anchor=center, inner sep=0}, draw=none, from=1-2, to=2-2]
	\arrow["\la", from=1-4, to=1-5]
	\arrow[""{name=1, anchor=center, inner sep=0}, "\epsilon"', from=1-4, to=2-4]
	\arrow[""{name=2, anchor=center, inner sep=0}, draw=none, from=1-6, to=2-6]
	\arrow["\la", from=1-8, to=1-9]
	\arrow[""{name=3, anchor=center, inner sep=0}, "\epsilon"', from=1-8, to=2-8]
	\arrow["\lb"', from=2-4, to=2-5]
	\arrow["\la"', color={rgb,255:red,214;green,92;blue,92}, from=2-5, to=2-6]
	\arrow["\epsilon"', from=2-6, to=1-5]
	\arrow["\lb"', from=2-8, to=2-9]
	\arrow["\la"', from=2-9, to=2-10]
	\arrow["\epsilon"', from=2-9, to=3-9]
	\arrow["\epsilon"', from=2-10, to=1-9]
	\arrow["\lb"', from=3-9, to=3-10]
	\arrow["\la"', color={rgb,255:red,214;green,92;blue,92}, from=3-10, to=3-11]
	\arrow["\epsilon"', from=3-11, to=2-10]
	\arrow["{T_0}", between={0.4}{0.6}, Rightarrow, from=0, to=1]
	\arrow["{T_0}", between={0.4}{0.6}, Rightarrow, from=2, to=3]
\end{tikzcd}\]
Since each step adds an $\la$-transition that needs patching, $T_0^{*}$ will be infinite.

There is also a symmetric counterexample: with $H := \{ \la\lb \leq \la \}$ applied to $\{\la\}$, the output should be $\la\lb^{*}$ but $T_0^{*}$ will again produce an infinite automaton.
\end{example}

To deal with these cases we will improve the construction.%
\footnote{%
    The reader may think to fix this by first ``saturating'' the hypothesis itself, to skip redundant steps; in the above case, obtaining $\lb^{*}\la \leq \la$.
    While this works in principle, it is circular: to calculate a hypothesis closed expression, we construct the automaton; to do that, we calculate a hypothesis closed expression.
    Instead, we seek to create a similar effect on the automaton level, where we can be certain it will terminate.
}
Examining the target expression $\lb^{*}\la$, our output should have a $\lb$-loop. 
Indeed, the (sub)word to be expanded, $\la$, appears in in the expansion $\lb\la$; on the automaton level, there is a state with $\la$ in its language. % chktex 36
So, rather than starting a new patch, we reset the current one to its initial state.
We call this an \emph{initial state reset}.
% https://q.uiver.app/#q=WzAsMTAsWzQsMCwiXFxjaXJjIl0sWzYsMCwiXFxvcGx1cyJdLFsyLDAsIlxcb3BsdXMiXSxbMCwwLCJcXGNpcmMiXSxbNCwxLCJcXGJ1bGxldCJdLFs1LDEsIlxcYnVsbGV0Il0sWzYsMSwiXFxidWxsZXQiXSxbMCwxLCJcXGJ1bGxldCJdLFsxLDEsIlxcYnVsbGV0Il0sWzIsMSwiXFxidWxsZXQiXSxbMCwxLCJcXGxhIl0sWzMsMiwiXFxsYSJdLFswLDQsIlxcZXBzaWxvbiIsMl0sWzQsNSwiXFxsYiIsMl0sWzUsNiwiXFxsYSIsMl0sWzYsMSwiXFxlcHNpbG9uIiwyXSxbMyw3LCJcXGVwc2lsb24iLDJdLFs5LDIsIlxcZXBzaWxvbiIsMl0sWzcsOCwiXFxsYiIsMl0sWzgsOSwiXFxsYSIsMl0sWzUsNCwiXFxlcHNpbG9uIiwyLHsiY3VydmUiOjEsInN0eWxlIjp7ImJvZHkiOnsibmFtZSI6ImRhc2hlZCJ9fX1dLFsxNywxMiwiIiwwLHsic2hvcnRlbiI6eyJzb3VyY2UiOjQwLCJ0YXJnZXQiOjQwfX1dXQ==
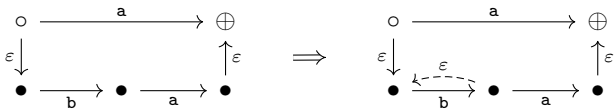
\[\begin{tikzcd}[row sep=scriptsize]
	\circ && \oplus && \circ && \oplus \\
	\bullet & \bullet & \bullet && \bullet & \bullet & \bullet
	\arrow["\la", from=1-1, to=1-3]
	\arrow["\epsilon"', from=1-1, to=2-1]
	\arrow["\la", from=1-5, to=1-7]
	\arrow[""{name=0, anchor=center, inner sep=0}, "\epsilon"', from=1-5, to=2-5]
	\arrow["\lb"', from=2-1, to=2-2]
	\arrow["\la"', from=2-2, to=2-3]
	\arrow[""{name=1, anchor=center, inner sep=0}, "\epsilon"', from=2-3, to=1-3]
	\arrow["\lb"', from=2-5, to=2-6]
	\arrow["\epsilon"', curve={height=6pt}, dashed, from=2-6, to=2-5]
	\arrow["\la"', from=2-6, to=2-7]
	\arrow["\epsilon"', from=2-7, to=1-7]
	\arrow[between={0.4}{0.6}, Rightarrow, from=1, to=0]
\end{tikzcd}\]
Using this new automaton, the automaton in \cref{ex:sat-motiv} needs only one patch.

For the symmetric hypothesis $\la\lb \leq \la$, we apply a similar tactic --- in this case we search for states that have an $\la$ path \emph{from the initial state}.
For such states introduce a \emph{final state reset}, back to that state from the final state.
% https://q.uiver.app/#q=WzAsMTAsWzIsMCwiXFxvcGx1cyJdLFswLDAsIlxcY2lyYyJdLFs0LDAsIlxcY2lyYyJdLFs0LDEsIlxcYnVsbGV0Il0sWzYsMCwiXFxvcGx1cyJdLFs2LDEsIlxcYnVsbGV0Il0sWzUsMSwiXFxidWxsZXQiXSxbMCwxLCJcXGJ1bGxldCJdLFsxLDEsIlxcYnVsbGV0Il0sWzIsMSwiXFxidWxsZXQiXSxbMSwwLCJcXGxhIl0sWzIsNCwiXFxsYSJdLFs1LDQsIlxcZXBzaWxvbiIsMl0sWzIsMywiXFxlcHNpbG9uIiwyXSxbMyw2LCJcXGxhIiwyXSxbNiw1LCJcXGxiIiwyXSxbNSw2LCJcXGVwc2lsb24iLDIseyJjdXJ2ZSI6MSwic3R5bGUiOnsiYm9keSI6eyJuYW1lIjoiZGFzaGVkIn19fV0sWzEsNywiXFxlcHNpbG9uIiwyXSxbOSwwLCJcXGVwc2lsb24iLDJdLFs3LDgsIlxcbGEiLDJdLFs4LDksIlxcbGIiLDJdLFsxOCwxMywiIiwyLHsic2hvcnRlbiI6eyJzb3VyY2UiOjQwLCJ0YXJnZXQiOjQwfX1dXQ==
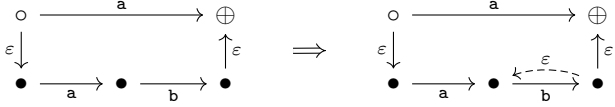
\[\begin{tikzcd}[row sep=scriptsize]
	\circ && \oplus && \circ && \oplus \\
	\bullet & \bullet & \bullet && \bullet & \bullet & \bullet
	\arrow["\la", from=1-1, to=1-3]
	\arrow["\epsilon"', from=1-1, to=2-1]
	\arrow["\la", from=1-5, to=1-7]
	\arrow[""{name=0, anchor=center, inner sep=0}, "\epsilon"', from=1-5, to=2-5]
	\arrow["\la"', from=2-1, to=2-2]
	\arrow["\lb"', from=2-2, to=2-3]
	\arrow[""{name=1, anchor=center, inner sep=0}, "\epsilon"', from=2-3, to=1-3]
	\arrow["\la"', from=2-5, to=2-6]
	\arrow["\lb"', from=2-6, to=2-7]
	\arrow["\epsilon"', from=2-7, to=1-7]
	\arrow["\epsilon"', curve={height=6pt}, dashed, from=2-7, to=2-6]
	\arrow[between={0.4}{0.6}, Rightarrow, from=1, to=0]
\end{tikzcd}\]
Both initial and final state resets prevent the need to patch a patched on automaton $\MZ$, when that work could have been done before even attaching $\MZ$ to $\MX$.
We now define the general case of adding these transitions, then use it in place of the automaton for $e$ when patching.

\begin{definition}[Saturated automaton]\label{def:sat-aut}
  Let $\MZ = (Z, \rightarrow, z_\oplus, z_0)$ be an automaton, and $w$ a word.
  The $w$-\emph{saturation} of $\MZ$, written $\MZ^w$, is $(Z, \to_S, z_\oplus, z_0)$, where $\to_S$ is the smallest subset of $Z \times (\Sigma \times \{\epsilon\}) \times Z$ satisfying the rules:
\begin{mathpar}
    \inferrule{%
        z \tr{a} z'
    }{%
        z \tr{a}_S z'
    } 
    \and
    \inferrule{%
        z \str[S]{w} z_\oplus
    }{%
        z \eptr_S z_0
    } 
    \and
    \inferrule{%
        z_0 \str[S]{w} z
    }{%
        z_\oplus \eptr_S z
    } 
\end{mathpar}
The latter rules correspond to initial and final state resets respectively.
\end{definition}

Saturation adds only transitions, so it will preserve finiteness of automata.
Note that the transition relation of $\MZ^w$ is well-defined, as adding a new $\epsilon$-transition does not prevent any other $\epsilon$-transition from being added.

\begin{example}\label{ex:sat}
We compute the $\la$-saturation of the automaton for $\la\lb + \lb\la$.
% https://q.uiver.app/#q=WzAsOCxbMCwwLCJcXGJ1bGxldCJdLFsxLDAsIlxcYnVsbGV0Il0sWzQsMCwiXFxidWxsZXQiXSxbNSwxLCJcXGJ1bGxldCJdLFs0LDEsIlxcYnVsbGV0Il0sWzUsMCwiXFxidWxsZXQiXSxbMCwxLCJcXGJ1bGxldCJdLFsxLDEsIlxcYnVsbGV0Il0sWzAsMSwiXFxsYSJdLFsyLDQsIlxcbGIiLDJdLFsyLDUsIlxcbGEiXSxbNSwzLCJcXGxiIiwyXSxbNCwzLCJcXGxhIiwyXSxbMyw1LCJcXGVwc2lsb24iLDIseyJjdXJ2ZSI6MSwiY29sb3VyIjpbMCw2MCw2MF19LFswLDYwLDYwLDFdXSxbNCwyLCJcXGVwc2lsb24iLDIseyJjdXJ2ZSI6MSwiY29sb3VyIjpbMCw2MCw2MF19LFswLDYwLDYwLDFdXSxbMCw2LCJcXGxiIiwyXSxbNiw3LCJcXGxhIiwyXSxbMSw3LCJcXGxiIl0sWzE3LDksIlxcdGV4dHsoJFxcbGEkLXNhdHVyYXRpb24pfSIsMCx7InNob3J0ZW4iOnsic291cmNlIjo0MCwidGFyZ2V0Ijo0MH19XV0=
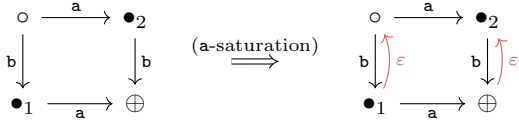
\[\begin{tikzcd}
	\circ & \bullet_2 &&& \circ & \bullet_2 \\
	\bullet_1 & \oplus &&& \bullet_{1} & \oplus
	\arrow["\la", from=1-1, to=1-2]
	\arrow["\lb"', from=1-1, to=2-1]
	\arrow[""{name=0, anchor=center, inner sep=0}, "\lb", from=1-2, to=2-2]
	\arrow["\la", from=1-5, to=1-6]
	\arrow[""{name=1, anchor=center, inner sep=0}, "\lb"', from=1-5, to=2-5]
	\arrow["\lb"', from=1-6, to=2-6]
	\arrow["\la"', from=2-1, to=2-2]
	\arrow["\epsilon"', color={rgb,255:red,214;green,92;blue,92}, curve={height=6pt}, from=2-5, to=1-5]
	\arrow["\la"', from=2-5, to=2-6]
	\arrow["\epsilon"', color={rgb,255:red,214;green,92;blue,92}, curve={height=6pt}, from=2-6, to=1-6]
	\arrow["{\text{($\la$-saturation)}}", between={0.4}{0.6}, Rightarrow, from=0, to=1]
\end{tikzcd}\]
The left \textcolor{q-red}{$\eptr$} is added because $\bullet_1 \tr{\la} \oplus$ in the original automaton; it is an initial state reset.
The right \textcolor{q-red}{$\eptr$} is added because $\circ \tr{\la} \bullet_{2}$; it is a final state reset.
\end{example}

Incorporating the saturated automaton into the existing construction gives a revised version of $T_0$, which patches using $\MZ^w$ in place of $\MZ$.

\begin{definition}[$T_H$]\label{def:TH}
Write $X = \{ x_0, \dots, x_{n-1} \}$; we define $T_H(\MX) = \MX_n$, with
\begin{mathpar}
    \MX_0 = \MX
    \and
    \MX_{i+1} = 
        \begin{cases}
        \MX_i & \quad w\inv l_\MX(x_i) \sube \sem{e}\inv l_\MX(x_i) \\
        \MX_i\satpat{x_i} & \quad \text{otherwise}
        \end{cases}
\end{mathpar}
Finally, $T_H^{*}(\MX)$ is the limit of the chain $\MX \sqsube T_H(\MX) \sqsube T_H^2(\MX) \sqsube \ldots$.
\end{definition}

\begin{definition}[Candidate partial reduction: $r_H$]\label{def:rH}
  Given a regular expression $g$, let $\MX$ be its finite automaton.
  If $T_H^*(\MX)$ is finite, then $r_H(g)$ is the regular expression for the initial state of $T_H^{*}(\MX)$.
  Otherwise, it is undefined.
\end{definition}

To show that $r_H$ is a partial reduction, we first show $\sem{g}_H \sube \sem{r_H(g)}$:

\begin{lemma}\label{lem:THcorrect}
For every $x \in X$, it holds that $H^{*}(l_\MX(x)) \sube l_{T_H^{*}(\MX)}(x)$.
\end{lemma}
\begin{proof}
The proof follows fairly directly from the corresponding proof for $T_0$ (\cref{lem:T0correct}).
There, we proved that $H(l_\MX(x)) \sube l_{T_0(\MX)}(x)$.
Since $\MZ \sqsubseteq \MZ^w$, we have $T_0(\MX) \sqsube T_H(\MX)$; therefore, for any $x \in X$, it follows that $l_{T_0(\MX)}(x) \sube l_{T_H(\MX)}(x)$.
Thus, the desired inclusion follows by transitivity.
\end{proof}

Now we verify that $r_H(g) \leqq_H g$.
Similarly to $T_0$, we leverage completeness of $\KA$ (see \cref{thm:kozen}).
However, the proof for $T_0$ required that the least solution to the patched on $\MZ$ at the initial state $z_0$ is equivalent to $e$.
In contrast, for $T_H$, the patched on automaton is $w$-saturated, so its least solution is generally above $e$.
This makes sense: the automaton was $w$-saturated so that it would accept more words, so its least solution \emph{should} grow w.r.t.\ $\leqq$.
However, it should \emph{not} change w.r.t.\ $\leqq_H$, which is exactly what we will show: though the least solution at the initial state is not equivalent to $e$, it is equivalent to $e$ \emph{up to} $H$.

The proof requires that we briefly introduce the \emph{reverse} operation on words, languages, etc.
We write $w^r$ for the reverse of a word $w$; this lifts to languages: $L^r = \{ w^r : w \in L \}$.
The reverse $g^r$ of an expression $g$ sends $g_1 \cdot g_2$ to $g_2^r \cdot g_1^r$, and acts homomorphically on the other operators, so $\sem{g^r} = \sem{g}^r$.
The laws of $\KA$ are unperturbed under reversal, so $\KA \vdash f = g$ iff $\KA \vdash f^r = g^r$.

The \emph{reverse automaton} $\MX^r$ refers to $\MX$ with all transitions reversed, and initial and final states swapped; clearly, $L_\MX^r = L_{\MX^r}$.
We can then speak of \emph{reverse solutions}, i.e., solutions to the reversed automaton $\MX^r$.
There is a tight relationship between solutions and reverse solutions at the initial and final states.

\begin{lemma}\label{lem:rev-init-soln}
If $s$ and $s^r$ are respectively the least and least reverse solution to some automaton $\MZ$, then $s^r(z_\oplus) \equiv {s(z_0)}^r$
\end{lemma}
\begin{proof}
Since $s$ is the least solution to $\MZ$, we know by \cref{lem:soln-language} that $\sem{s(z_0)} = L_\MZ$; similarly, $\sem{s^r(z_\oplus)} = L_{\MZ^r}$. 
With the observations above, we then derive:
\[
    \sem{s^r(z_\oplus)} = L_{\MZ^r} = L_\MZ^r = \sem{s(z_0)}^r = \sem{{s(z_0)}^r}
\]
The claim follows from the above by \cref{thm:kozen}.
\end{proof}

\begin{restatable}{lemma}{satPresInitSoln}\label{lem:satPresInitSoln}
  Suppose that $\MZ$ is an automaton representing the expression $e$.
  Let $s^{+}$ be the least solution to the $w$-saturated automaton $\MZ^w$.
  Then $e \equiv_H s^{+}(z_0)$.
\end{restatable}
\begin{proof}[Proof Sketch]
$\MZ^w$ can be constructed by adding one $\epsilon$-transition to $\MZ$ at a time.
One can then show that when an initial state reset is added, solutions are preserved, but only up to $H$.
Similarly, when a final state reset is added, \emph{reverse} solutions are preserved, up to $H^r$, the reveresed hypothesis.

\cref{lem:rev-init-soln} is then applied to relate the least reverse solution at the final state to the reverse of the least solution at the initial state, with its result expanded to $H$ and $H^r$-least solutions.
Therefore, at each step the least solution at the initial state does not change, up to $H$.
\end{proof}

Now we can replay the proof strategy for \cref{lem:extend-soln-T0} to obtain a corresponding lemma that $T_H$ preserves $H$-solutions.

\begin{restatable}{lemma}{extendSolnTH}\label{lem:extend-soln-TH}
Suppose $T_H^{*}(\MX)$ is finite, and $s$ is a solution to $\MX$.
We can construct an $H$-solution $s^{*}$ to $T_H^{*}(\MX)$, which moreover agrees with $s$ on all $x \in X$.
\end{restatable}
\begin{proof}[Proof Sketch]
Define $s^{*}$ as in the proof of \cref{lem:extend-soln-T0}, but for states in $\MZ^w_i$ use $s^w$, the least solution to $\MZ^w$, rather than $s_{\MZ}$.
This map satisfies most conditions to be a solution of $T_H(\MX)$, in particular those in $\MZ^w_i$, even for reset transitions.

It remains to validate the conditions originating from $\epsilon$-transitions into and out of the patched on automata, and these can be solved in the same way as in \cref{lem:extend-soln-T0} --- with one alteration.
To show that $s^{*}(z_{0}^i) \leqq s^{*}(x_i)$ for each $x_i$, instead of using $s_{\MZ}(z_0^i) \equiv e$, we now use that $s_{\MZ^w}(z_0^i) \equiv_H e$ (per \cref{lem:rev-init-soln}).
\end{proof}

The proof that $r_H$ is a partial reduction is effectively the same as that of \cref{thm:partialRedTZ}, using analogous lemmas about $T_H$, so we do not spell it out here.

\begin{restatable}{theorem}{partialRedTH}\label{thm:partialRedTH}
$r_H$ is a partial reduction.
\end{restatable}

Using \cref{lem:reduction-complete} we then conclude decidability and completeness for expressions in the domain of $r_H$, i.e., expressions where $T_H^{*}$ outputs a finite automaton.

\begin{corollary}\label{cor:partcompTH}
  Let $g,h$ be expressions for which $r_H$ is defined.
  If $\sem{g}_H = \sem{h}_H$, then $g \equiv_H h$; moreover, the latter is decidable.
  Thus, if $r_H$ is known to output finite automata for all expressions in $\Exp$, then $\KA_H$ is complete and decidable.
\end{corollary}

$T_0^{*}$ never outputs a finite automaton when $T_H^{*}$ does not, and $T_H^{*}$ outputs finitely for the hypothesis/expression pairs from \cref{ex:sat-motiv}.
But does $T_0^*$ output a finite automaton \emph{when it can?}
Unfortunately, this is not the case.

\begin{example}\label{ex:TH-nonterm}
  Let $H := \{ \la\lb \leq \lb\la \}$, and consider $\lb\la^{*}$.
  The hypothesis closure of this expression is $\la^{*}\lb\la^{*}$.
  $T_H^{*}$ proceeds infinitely, as sketched below:
  % https://q.uiver.app/#q=WzAsMTUsWzAsMCwiXFxjaXJjIixbMCw2MCw2MCwxXV0sWzAsMiwiXFxvcGx1cyIsWzAsNjAsNjAsMV1dLFsyLDAsIlxcY2lyYyJdLFsyLDIsIlxcb3BsdXMiLFswLDYwLDYwLDFdXSxbMywwLCJcXGJ1bGxldCJdLFszLDEsIlxcYnVsbGV0IixbMCw2MCw2MCwxXV0sWzMsMiwiXFxidWxsZXQiLFswLDYwLDYwLDFdXSxbNSwwLCJcXGNpcmMiXSxbNSwyLCJcXG9wbHVzIixbMCw2MCw2MCwxXV0sWzYsMCwiXFxidWxsZXQiXSxbNiwxLCJcXGJ1bGxldCJdLFs2LDIsIlxcYnVsbGV0Il0sWzcsMSwiXFxidWxsZXQiXSxbNywyLCJcXGJ1bGxldCIsWzAsNjAsNjAsMV1dLFs3LDMsIlxcYnVsbGV0IixbMCw2MCw2MCwxXV0sWzAsMSwiXFxsYiIsMCx7ImNvbG91ciI6WzAsNjAsNjBdfSxbMCw2MCw2MCwxXV0sWzEsMSwiXFxsYSIsMCx7ImFuZ2xlIjotOTAsImNvbG91ciI6WzAsNjAsNjBdfSxbMCw2MCw2MCwxXV0sWzIsMywiXFxsYiJdLFszLDMsIlxcbGEiLDAseyJhbmdsZSI6LTkwLCJjb2xvdXIiOlswLDYwLDYwXX0sWzAsNjAsNjAsMV1dLFsyLDQsIlxcZXBzaWxvbiJdLFs0LDUsIlxcbGEiXSxbNSw2LCJcXGxiIiwwLHsiY29sb3VyIjpbMCw2MCw2MF19LFswLDYwLDYwLDFdXSxbNiwzLCJcXGVwc2lsb24iLDIseyJjb2xvdXIiOlswLDYwLDYwXX0sWzAsNjAsNjAsMV1dLFs3LDksIlxcZXBzaWxvbiJdLFs5LDEwLCJcXGxhIl0sWzExLDgsIlxcZXBzaWxvbiIsMl0sWzcsOCwiXFxsYiIsMl0sWzgsOCwiXFxsYSIsMCx7ImFuZ2xlIjotOTAsImNvbG91ciI6WzAsNjAsNjBdfSxbMCw2MCw2MCwxXV0sWzEwLDEyLCJcXGVwc2lsb24iXSxbMTIsMTMsIlxcbGEiXSxbMTMsMTQsIlxcbGIiLDAseyJjb2xvdXIiOlswLDYwLDYwXX0sWzAsNjAsNjAsMV1dLFsxNCw4LCJcXGVwc2lsb24iLDAseyJjdXJ2ZSI6LTIsImNvbG91ciI6WzAsNjAsNjBdfSxbMCw2MCw2MCwxXV0sWzEwLDExLCJcXGxiIl0sWzQsNiwiIiwwLHsic3R5bGUiOnsiYm9keSI6eyJuYW1lIjoibm9uZSJ9LCJoZWFkIjp7Im5hbWUiOiJub25lIn19fV0sWzE1LDE3LCJUX0giLDAseyJzaG9ydGVuIjp7InNvdXJjZSI6NDAsInRhcmdldCI6NDB9fV0sWzMzLDI2LCJUX0giLDAseyJzaG9ydGVuIjp7InNvdXJjZSI6NDAsInRhcmdldCI6NDB9fV1d
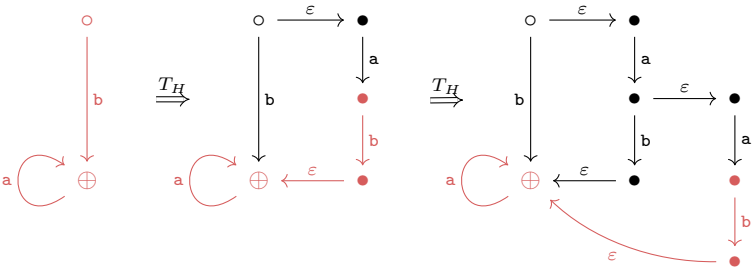
\[\begin{tikzcd}
	\textcolor{rgb,255:red,214;green,92;blue,92}{\circ} && \circ & \bullet && \circ & \bullet \\
	&&& \textcolor{rgb,255:red,214;green,92;blue,92}{\bullet} &&& \bullet & \bullet \\
	\textcolor{rgb,255:red,214;green,92;blue,92}{\oplus} && \textcolor{rgb,255:red,214;green,92;blue,92}{\oplus} & \textcolor{rgb,255:red,214;green,92;blue,92}{\bullet} && \textcolor{rgb,255:red,214;green,92;blue,92}{\oplus} & \bullet & \textcolor{rgb,255:red,214;green,92;blue,92}{\bullet} \\
	&&&&&&& \textcolor{rgb,255:red,214;green,92;blue,92}{\bullet}
	\arrow[""{name=0, anchor=center, inner sep=0}, "\lb", color={rgb,255:red,214;green,92;blue,92}, from=1-1, to=3-1]
	\arrow["\epsilon", from=1-3, to=1-4]
	\arrow[""{name=1, anchor=center, inner sep=0}, "\lb", from=1-3, to=3-3]
	\arrow["\la", from=1-4, to=2-4]
	\arrow[""{name=2, anchor=center, inner sep=0}, draw=none, from=1-4, to=3-4]
	\arrow["\epsilon", from=1-6, to=1-7]
	\arrow[""{name=3, anchor=center, inner sep=0}, "\lb"', from=1-6, to=3-6]
	\arrow["\la", from=1-7, to=2-7]
	\arrow["\lb", color={rgb,255:red,214;green,92;blue,92}, from=2-4, to=3-4]
	\arrow["\epsilon", from=2-7, to=2-8]
	\arrow["\lb", from=2-7, to=3-7]
	\arrow["\la", from=2-8, to=3-8]
	\arrow["\la", color={rgb,255:red,214;green,92;blue,92}, from=3-1, to=3-1, loop, in=145, out=215, distance=10mm]
	\arrow["\la", color={rgb,255:red,214;green,92;blue,92}, from=3-3, to=3-3, loop, in=145, out=215, distance=10mm]
	\arrow["\epsilon"', color={rgb,255:red,214;green,92;blue,92}, from=3-4, to=3-3]
	\arrow["\la", color={rgb,255:red,214;green,92;blue,92}, from=3-6, to=3-6, loop, in=145, out=215, distance=10mm]
	\arrow["\epsilon"', from=3-7, to=3-6]
	\arrow["\lb", color={rgb,255:red,214;green,92;blue,92}, from=3-8, to=4-8]
	\arrow["\epsilon", color={rgb,255:red,214;green,92;blue,92}, curve={height=-12pt}, from=4-8, to=3-6]
	\arrow["{T_H}", between={0.4}{0.6}, Rightarrow, from=0, to=1]
	\arrow["{T_H}", between={0.4}{0.6}, Rightarrow, from=2, to=3]
\end{tikzcd}\]
Since $\lb\la$ does not appear as a path in the patched-on automaton, there is nothing to saturate.
Indeed, it is only \emph{in the context of $\MX$} that the need to introduce a loop becomes clear.
This separates this example from saturation as an operation independent from patching.
A further alteration to $T_H$ to deal with this sort of scenario is an exciting endeavour, but is left for now to future work.
\end{example}

Similar to $T_0^{*}$, $T_H^{*}$ can be applied to ``independent'' sets of hypotheses.
There is already significant work towards creating ``modular'' completeness proofs by composing reductions~\cite{tools}.
However, this work is orthogonal to ours --- composing reductions benefits from a formula for creating new ones, which $T_H^{*}$ offers.

These techniques can also be used to construct reductions to \emph{any system that reduces to} $\KA$, via an existing reduction.
For example, as per~\cite{tools}, both $\KAT$ and $\netkat$ can be regarded as $\KA$ with infinitely many hypotheses.
Reductions introduced there can be combined with one produced by $T_H^{*}$ to show completeness of $\KAT$ with some hypotheses added, or decide equivalence of two $\netkat$-expressions with hypotheses.

\section{Conclusion \& Further Work}\label{sec:conclusion}

We have presented a general strategy to construct reductions for Kleene algebra with a linear hypothesis, through a construction on automata.
Decidability and completeness of the relevant system follow for the cases where the construction terminates.
This strategy can then be used for sets of independent linear hypotheses as well, and recover known reductions in the literature~\cite{ckao,tools}.

Hypothesis closure provides a useful canonical language semantics, but they are no substitute for the insight that a ``hand-crafted'' semantics provides.
It remains an important task to develop these for sufficiently ubiquitous systems.

There are also sets of hypotheses that the tools given here do not cover, namely those where hypotheses directly interact with one another.
Another difficulty is with general hypotheses of the form $e \leq f$, where $f$ is some general regular expression.
It is a natural direction for future work to expand the automaton construction to deal with these uncovered cases.
This would allow for creating reductions to $\KA$, or other known-complete systems like $\KAT$, in a wider array of Kleene algebras with hypotheses.

Another direct line of future work is studying when the construction ``succeeds'', i.e., outputs a finite automaton; or indeed, when that is even possible. Characterising when the construction can succeed would allow for more precise decidability and completeness results.
Commutativity hypotheses in particular have already seen investigation, in~\cite{maarand-uustalu}, where they allow for modelling ``parallel actions'' in the style of Mazurkiewicz traces.

Finally, a word on a possible application of this work.
In~\cite{braibant-pous}, Braibant and Pous develop a tactic in Rocq to automatically prove equivalence of regular expressions using completeness and decidability of \KA.
The techniques introduced in this paper could be used to expand this tactic, where the user designates relevant facts in the context as hypotheses.

\begin{credits}
\subsubsection{\ackname}
The authors want to thank Yde Venema for valuable comments early on, and Bas Laarakker and Sean Prendi for proofreading.
This work was supported by the Dutch research council (NWO) under grant no.\ VI.Veni.232.286 (ChEOpS).
\end{credits}

\newpage

\bibliographystyle{splncs04}
\bibliography{bibliography.bib}

\ifarxiv%

\clearpage

\appendix
\section{Proof of \cref{ex:comm-while}}\label{app:example}

\noindent\emph{\cref{ex:comm-while}.} Program 3 and Program 4 are inequivalent in standard regular language semantics, but equivalent with the hypotheses:
\[ H = \{ \pb \cdot \pa = \pa \cdot \pb,\ \pb \cdot \test_1 = \test_1 \cdot \pb,\ \pb \cdot \overline{\test_{1}} = \overline{\test_{1}} \cdot \pb \}.\]
More precisely,
\[ \sem{
    \pb \cdot {(\test_1 \cdot \pa)}^{*} \cdot \overline{\test_{1}}
  }
  \neq
  \sem{
    {(\test_1 \cdot \pa)}^{*} \cdot \overline{\test_{1}} \cdot \pb.
  } \]
and therefore, by soundness of \KA,
\[
  \pb \cdot {(\test_1 \cdot \pa)}^{*} \cdot \overline{\test_{1}}
  \not\equiv
   {(\test_1 \cdot \pa)}^{*} \cdot \overline{\test_{1}} \cdot \pb.
\]
Where on the other hand, using the hypotheses, 
\[
  \pb \cdot {(\test_1 \cdot \pa)}^{*} \cdot \overline{\test_{1}}
  \equiv_{H}
   {(\test_1 \cdot \pa)}^{*} \cdot \overline{\test_{1}} \cdot \pb.
\]

\begin{proof}

  It's clear to see that the two expressions have different regular language semantics, because for example, the word $\pb \test_1 \pa \overline{\test_1}$ is in the semantics of the left expression, and not of the semantics of the right one.

  On the other hand, to prove that the two expressions are equivalent, we cannot rely on completeness of \KA (because we have added extra hypotheses) and so we must prove it directly. To this end we claim that both expressions are equivalent to an intermediate one:
  \begin{align*}
    \pb \cdot \aloop \cdot \overline{\test_1} &\equiv \aloop \cdot \pb \cdot \aloop \cdot \overline{\test_1} \\
    \aloop \cdot \overline{\test_1} \cdot \pb &\equiv \aloop \cdot \pb \cdot \aloop \cdot \overline{\test_1}
  \end{align*}
We already apply our hypothesis $\pb \cdot \overline{\test_1} = \overline{\test_1} \cdot \pb$, and observe that proving the above equivalences is equivalent to proving the following, by just stripping $\overline{\test_1}$ off the right side of all expressions:
  \begin{align*}
    \pb \cdot \aloop &\equiv \aloop \cdot \pb \cdot \aloop \\
    \aloop \cdot \pb &\equiv \aloop \cdot \pb \cdot \aloop
  \end{align*}
  Note that since $x^{*} \equiv 1 + xx^{*}$ for all $x$, it follows that $1 \leqq x^{*}$ and therefore $1 \leqq \aloop$. In this way, the $\leqq$ direction of both equivalences follows immediately. It remains to show that:
  \begin{align*}
     \aloop \cdot \pb \cdot \aloop &\leqq \pb \cdot \aloop\\
     \aloop \cdot \pb \cdot \aloop &\leqq \aloop \cdot \pb.
  \end{align*}
  We begin with the former. Recall the following \KA axiom:
  \[ e + f \cdot g \leq g \Longrightarrow f^{*} \cdot e \leq g. \]
  If we identify $f = (\test_1 \cdot \pa)$ and $e,g = \pb \cdot \aloop$, we see that to prove the above it suffices to show that:
  \[ \pb \cdot \aloop + (\test_1 \cdot \pa) \cdot \pb \cdot \aloop \leqq \pb \cdot \aloop\]
  By idempotence, this is equivalent to showing that:
  \[ (\test_1 \cdot \pa) \cdot \pb \cdot \aloop \leqq \pb \cdot \aloop.\]
  To show this, we reason:
  \begin{align*}
    (\test_1 \cdot \pa) \cdot \pb \cdot \aloop &\equiv_H \pb \cdot (\test_1 \cdot \pa) \cdot \aloop \\
    &\leqq\hspace{.3cm} \pb \cdot \aloop
  \end{align*}
  where the first step follows by the hypotheses that $\pb$ commutes with the test $\test_1$ and the action $\pa$. The inequality follows because $\aloop \equiv 1 + (\test_1 \cdot \pa) \cdot \aloop$ (axiom), and so we have $(\test_1 \cdot \pa) \cdot \aloop \leqq \aloop$.

  Now it remains to prove:
  \[ \aloop \cdot \pb \cdot \aloop \leqq \aloop \cdot \pb \]
  The proof proceeds in a symmetric fashion. Recall the \KA axiom:
  \[ e + f \cdot g \leq f \Longrightarrow e \cdot g^{*} \leq f. \]
  Identifying $g = (\test_1 \cdot \pa)$ and $e,f = \aloop \cdot \pb$, it suffices to prove:
  \[ \aloop \cdot \pb + \aloop \cdot \pb \cdot (\test_1 \cdot \pa) \leq \aloop \cdot \pb. \]
  By idempotency, this is equivalent to proving:
  \[ \aloop \cdot \pb \cdot (\test_1 \cdot \pa) \leq \aloop \cdot \pb. \]
  We recall the fact that $(\test_1 \cdot \pa) \cdot \aloop \leqq \aloop$ and complete the proof:
  \begin{align*}
    \aloop \cdot \pb \cdot (\test_1 \cdot \pa) &\equiv_H \aloop \cdot (\test_1 \cdot \pa) \cdot \pb \\
    &\leqq\hspace{.3cm} \aloop \cdot \pb.
  \end{align*}
\end{proof}
\newpage

\section{Detailed Proofs}\label{app:proofs}

We now present proofs of the lemmas from the main text, and remind the reader of the assumption below, made at the beginning of \cref{sec:patching}.
\bigAssumption*

\termCond*
\begin{proof}
Suppose that for every $x \in X$, $w\inv l_{\MX}(x) \sube \sem{e}\inv l_{\MX}(x)$. 
Let $x \in X$; we want to show that $l_{\MX}(x) = H(l_{\MX}(x))$.
The $\sube$ direction follows by the definition of $H$.
For the $\supe$ direction, we instantiate \cref{def:hyp-cl-lang}:
\[ H(l_{\MX}(x)) := l_{\MX}(x) \cup \bigcup \{ u \sem{e} v : u w v \in l_{\MX}(x) \} \]
so it suffices to show that if $u w v \in l_{\MX}(x)$, then $u \sem{e} v \sube l_{\MX}(x)$.
To this end, take $uw'v$ where $w' \in \sem{e}$ and $uwv \in l_{\MX}(x)$.
Let $x'$ be a state reached after $u$ has been read in any accepting trace of $uwv$.
Naturally, $wv \in l_{\MX}(x')$, and so $v \in w\inv l_{\MX}(x')$.
Because $w\inv l_{\MX}(x') \sube \sem{e}\inv l_{\MX}(x')$, we have $v \in \sem{e}\inv l_{\MX}(x')$, whence $w'v \in l_{\MX}(x')$.
Since we could read $u$ while moving from $x$ to $x'$, $uw'v \in l_{\MX}(x)$.

For the other direction, suppose that for every $x \in X$, $l_{\MX}(x) = H(l_{\MX}(x))$.
We want to show that $w\inv l_{\MX}(x) \sube \sem{e}\inv l_{\MX}(x)$. 
Let $v \in w\inv l_{\MX}(x)$, so $wv \in l_{\MX}(x)$.
By the definition of $H$, $\sem{e} v \sube H(l_{\MX}(x))$.
But we assumed that $H(l_{\MX}(x)) = l_{\MX}(x)$, so $\sem{e} v \sube l_{\MX}(x)$, and finally $v \in \sem{e}\inv l_{\MX}(x)$.
\end{proof}

\TZcorrect*
\begin{proof}
It suffices to prove the claim for one application of $H$ and $T_0$, i.e.:
\[ H(l_\MX(x)) \sube l_{T_0(\MX)}(x). \]
To this end, we prove that when $uwv \in l_\MX(x)$ and $w' \in \sem{e}$, $uw'v \in l_{T_0(\MX)}(x)$.
If $uw'v \in l_{\MX}(x)$ then we're done, so assume it is not.
We dissect the trace witnessing that $uwv \in l_\MX(x)$ as follows:
\[ x \str{u} y \str{w} y' \str{v} x_\oplus. \]
Note that $wv \in l_\MX(y)$, and by our assumption above, $w'v \not\in l_\MX(y)$ --- otherwise we could sequence it after $u$ and have an $\MX$-trace from $x$.
This tells us that $l_{\MX}(y)$ is not $H$-closed, and thus, by \cref{lem:term-cond}, $w\inv l_\MX(y) \not\sube \sem{e}\inv l_\MX(y)$.
This being the case, in $T_0(\MX)$ the state $y$ will be $w$-patched.
That is, for some $i$, $\MX_{i+1} = \MX_i\patch{y}$.
We can now produce a new $\MX_{i+1}$ trace:
\[ x \str{u} y \eptr z_0^i \str{w'} z_{\oplus}^i \eptr  y' \str{v} x_\oplus \]
to see that $uw'v \in l_{\MX_{i+1}}(x)$.
Finally, we observe that patching is inflationary with respect to $\sqsube$, since it only adds states and transitions.
Therefore $\MX_i \sqsube \MX_n = T_0(\MX)$.
As a result, $uw'v \in l_{T_0(\MX)}(x)$, and we are done.
\end{proof}

\solnAutExtended*
\begin{proof}
Let $s$ be a (least) solution to $\MX$.
For $s$ to be an $H$-solution to $\MX$, it must satisfy the same equations, but with $\leqq_H$ in place of $\leqq$.
Since $\KA_H$ contains all of the axioms of $\KA$, this is easily seen to be the case.

To show that any least solution is also a least $H$-solution, we refer to the proof of \cref{lem:least-soln-construct} as seen in e.g.~\cite{kozen-thm,conway-alg-aut}.
In this proof, a solution is constructed, and shown to be least.
We can use the same strategy with $\leqq_H$ in place of $\leqq$ to construct a least $H$-solution.
However, in doing so, the very same map from $X$ to $\Exp$ will be produced.
Therefore, this least solution is a least $H$-solution.
Similar proofs can also be found in~\cite[Lemma~3.12]{cka-fm}.
\end{proof}

\extendSolnTZ*
\begin{proof}
Since $T_0^{*}(\MX)$ is a finite automaton, $T_0^{*}(\MX) = T_0^n(\MX)$ for some $n \in \N$.
Therefore, we only concern ourselves with \emph{one} application of $T_0$, and we can re-apply the result to each successive operation of $T_0$.

By definition, $T_0(\MX)$ consists of $\MX$ with finitely many copies of the automaton $\MZ$, representing $\sem{e}$, patched on.
We will call these automata $\MZ_1, \ldots \MZ_n$ and the states to which they are $w$-patched $x_1, \ldots x_n$.
The state set $T_0(X)$ partitions into $X \cup Z_1 \cup \cdots \cup Z_n$.

All states in $\MX$ except $x_1, \dots, x_n$ have transitions to other states in $\MX$ exclusively, whereas the $x_i$ also transition to their corresponding $z_0^i$.
Similarly, the states of $\MZ_i$ only have outgoing transitions to other $\MZ_i$ states (with the same $i$), except for the $z_\oplus^i$ which also has an $\epsilon$-transition back to all $x_i'$ such that $x_i \str{w} x_i'$.
In total, $s': X \cup Z_1 \cup \cdots \cup Z_n \to \Exp$ is an $H$-solution to $T_0(\MX)$ if it satisfies:
\begin{mathpar}
    \inferrule{~}{%
        1 \leqq_H s'(x_\oplus)
    }
    \and
    \inferrule{%
        x \tr[\MX]{a} x'
    }{%
        a \cdot s'(x') \leqq_H s'(x)
    }
    \and
    \inferrule{%
        z \tr[\MZ_i]{a} z'
    }{%
        a \cdot s'(z') \leqq_H s'(z)
    }
    \\
    \inferrule{%
      i \leq n
    }{%
       s'(z_0^i) \leqq_H s'(x_i)
    }
    \and
    \inferrule{%
      x_i \str[\MX]{w} x'
    }{%
       s'(x') \leqq_H s'(z_\oplus^i)
    }
\end{mathpar}
Where $a \in \Sigma \cup 1$, implicitly interpreting 1 as $\epsilon$ in transitions.
We will now extend $s$, the least solution to $\MX$, into an $H$-solution for $T_0(\MX)$.
Let $s_{\MZ}$ be the least solution to $\MZ$, and thus $\MZ_i$ for every $i \leq n$.
We define $s^{*}: X \cup Z_1 \cup \cdots \cup Z_n \to \Exp$ as follows:
\[
s^{*}(y) = \begin{cases}
     s(y) & y \in X \\
     s_{\MZ}(y) \cdot r_i & y \in Z_i
\end{cases}
\quad\quad\text{where}\quad\quad
r_i := \sum_{x_i \str[\MX]{w} x' }s(x')
\]
By definition, $s^{*}$ agrees with $s$ for all $x \in X$, so it remains to show that $s^*$ satisfies the rules above.
Since $s^{*}$ coincides with $s$ on all states in $\MX$, the first two rules hold for $\leqq$, and hence for $\leqq_H$.
The third rule comes down to showing that, if $z \str[\MZ_i]{a} z'$, then $a \cdot s_\MZ(z') \cdot r_i \leqq_H s_\MZ(z) \cdot r_i$, which follows by the fact that $s_\MZ$ is a solution to $\MZ_i$, and monotonicity of sequential composition.

It remains to show that, for each $1 \leq i \leq n$, we have $s^*(z_0^i) \leqq_H s^*(x_i)$.
By definition of $s^*$ and $r_i$ it suffices to show that, for any $i$ and $x'$ with $x_i \str[\MX]{w} x'$:
\[ s_{\MZ_i}(z_0^i) \cdot s(x') \leqq_H s(x_i). \]
Now since $\MZ_i$ is an automaton representing $\sem{e}$ and $z_0^i$ is its initial state, $s_{\MZ_i}(z_0^i) \equiv e$ (by \Cref{thm:kozen}).
Applying this fact and our hypothesis $e \leq w$,
\[ s_{\MZ_i}(z_0^i) \cdot s(x') \leqq e \cdot s(x') \leqq_H w \cdot s(x') \leqq s(x_i) \]
where the final containment can be argued using the fact that $s$ is a solution to $\MX$, by induction on the construction of the trace $x_i \str[\MZ]{w} x'$.
\end{proof}

\satPresInitSoln*
\begin{proof}
Suppose that $s$ is an $H$-solution to $\MZ$ satisfying $s(z_0) \equiv_H e$.
Let $\MZ_f$ be an automaton obtained by adding an initial state reset to some state $z$ in $\MZ$.
Then we claim that $s$ is an $H$-solution to $\MZ_f$ as well.
Given that $\MZ_f$ is just $\MZ$ with one more $\epsilon$-transition, the majority of the necessary equations are satisfied immediately.
The only non-trivial equation comes from the new $\epsilon$-transition:
\[ s(z_0) \leqq_H s(z). \]
Because the $\epsilon$-transition $z \eptr z_0$ was added as an initial state reset, $z \str{w} z_\oplus$.
Since $s$ is an $H$-solution, it follows (by an easy induction) that $w \leqq_H s(z)$.
Stringing these together with our hypothesis $e \leq w$, we see that:
\[ s(z_0) \equiv_H e \leqq_H w \leqq_H s(z) \]
as desired.
So adding an initial state reset preserves the $H$-solutions of $\MZ$.

Now we make the symmetric claim: if $s^r$ is a reverse $H^r$-solution to $\MZ$ satisfying $s^r(z_\oplus) \equiv_{H^r} e^r$ where $H^r = \{ e^r \leq w^r \}$, and $\MZ_b$ is an automaton obtained by adding a final state reset to some state $z$ in $\MZ$, then $s^r$ is a reverse $H^r$-solution to $\MZ_b$ as well.
As with the first case, the majority of the requiremed equations are satisfied by definition.
We are left needing to prove:
\[ s^r(z_\oplus) \leqq_{H^r} s^r(z) \]
and we prove it in a symmetric way.
Since the final state reset $z_{\oplus} \eptr z$ was added, $z_0 \str{w} z$, so $w^r \leqq_{H^r} s^r(z)$.
We then derive as follows:
\[ s^r(z_\oplus) \equiv_{H^r} e^r \leqq_{H^r} w^r \leqq_{H^r} s^r(z). \]
So adding a final state reset preserves \emph{reverse} $H^r$-solutions for the automaton $\MZ$.
We now combine these two results to prove the lemma.

Our starting point is $\MZ$, an automaton representing $e$.
Let $\MZ_0 := \MZ$, and $\MZ_1, \MZ_2, \ldots \MZ_n$ be individual steps of adding either type of reset transition so that $\MZ_n = \MZ^w$.
Let $s_0, s_1, s_2, \ldots s_n$ be the least solutions to these automata, so $s_n = s^{+}$.
We claim that for all $i \leq n$, $s_i(z_0) \equiv_H e$.
We proceed by induction.
The base case is clear: $s_0$ is the least solution to $\MZ_0 = \MZ$ which is by definition the automaton representing the expression $e$, so by \cref{lem:soln-language} $\sem{s(z_0)} = \sem{e}$ and by \cref{thm:kozen}, $s(z_0) \equiv e$, allowing us to conclude that $s(z_0) \equiv_H e$. 

For the inductive step, we assume that $e \equiv_H s_i(z_0)$ where $s_i$ is the least solution to $\MZ_i$.
Now to obtain $\MZ_{i+1}$, we either add an initial state reset or a final state reset.
Suppose that we added an initial state reset.
Then we proved above that this means $s_i$ is also an $H$-solution to $\MZ_{i+1}$.
Therefore $s_{i+1}(z_0) \leqq_H s_i(z_0)$.
Since $\MZ_i \sqsubseteq \MZ_{i+1}$, we also have $s_i(z_0) \leqq_H s_{i+1}(z_0)$, meaning that:
\[ s_{i+1}(z_0) \equiv_H s_i(z_0) \equiv_H e. \]
Now suppose to get from $\MZ_i$ to $\MZ_{i+1}$, we added a final state reset.
Let $s_i^r, s_{i+1}^r$ be the least reverse solutions to these automata.
We use reasoning symmetric to the previous case to find that $s_{i+1}^r(z_\oplus) \equiv_{H^r} s_i^r(z_\oplus)$.
Taking a proof of this equivalence, we can then revert all of the expressions (including the hypotheses) to find that ${s_{i+1}^r(z_\oplus)}^r \equiv_H {s_i^r(z_\oplus)}^r$.
Using \cref{lem:rev-init-soln}, we can then conclude:
\[ s_{i+1}(z_0) \equiv {s_{i+1}^r(z_\oplus)}^r \equiv_H {s_i^r(z_\oplus)}^r \equiv s_i(z_0) \equiv_H e. \qedhere \]
\end{proof}

\extendSolnTH*
\begin{proof}
The proof largely follows the structure of \cref{lem:extend-soln-T0}, but with $T_H$ instead of $T_0$, and using $s_{\MZ_i^w}$ (the least solution to $\MZ^w_i$) rather than the least solution to $\MZ_i$.
At the very end, where an appeal to \cref{thm:kozen} is made to use that $s_{\MZ_i}(z_0^i) \equiv e$, we instead use \cref{lem:satPresInitSoln} to find that $s_{\MZ_i^w}(z_0^i) \equiv_H e$; this also suffices to make the final derivation go through.
\end{proof}

\partialRedTH*
\begin{proof}
  To show that $r_H$ is a partial reduction, we must show for all $g \in \Exp$ on which it is defined the following two properties hold:
  \begin{mathpar}
    \sem{g}_{H} \sube \sem{r_H(g)} \and r_H(g) \leqq_H g
  \end{mathpar}
The former follows from \cref{lem:THcorrect}: if $\MX$ is the automaton representing $g$, then $T_H^{*}(\MX)$ is finite (since $r_H$ is defined on $g$), and:
\[ \sem{g}_{H} = H^*(\sem{g}) = H^{*}(L_{\MX}) \sube L_{T_H^{*}(\MX)} = \sem{r_H(g)}. \]
For $r_H(g) \leqq_H g$, let $s$ be the least solution to $\MX$.
By \Cref{thm:kozen}, $s(x_0) \equiv g$.
Similarly let $s'$ be the least solution to $T_H^{*}(\MX)$; by definition, $s'(x_0) = r_H(g)$.
\cref{lem:extend-soln-T0} gives an $H$-solution $s^*$ to $T_H^{*}(\MX)$ such that $s^*(x) = s(x)$ for every $x \in X$.
Since $s'$ is the \emph{least} solution to $T_H^{*}(\MX)$, it is a least $H$-solution as well:
\[ r_H(g) = s'(x_0) \leqq_{H} s^*(x_0) = s(x_0) \equiv g. \qedhere \]
\end{proof}

\fi

\end{document}